\documentclass[letterpaper,twocolumn,10pt]{article}
\usepackage{usenix2019,epsfig,endnotes}

\usepackage{graphicx}
\usepackage[utf8]{inputenc}

\usepackage[tight,footnotesize]{subfigure}
\usepackage[english]{babel}
\usepackage{amsthm}

\usepackage{enumitem}
\usepackage{multirow}
\usepackage[unicode=true,
            colorlinks=true,
            linkcolor=blue]{hyperref}

\DeclareFixedFont{\ttb}{T1}{txtt}{bx}{n}{9} % for bold
\DeclareFixedFont{\ttm}{T1}{txtt}{m}{n}{9}  % for normal

\usepackage{color}
\definecolor{deepblue}{rgb}{0,0,0.5}
\definecolor{deepred}{rgb}{0.6,0,0}
\definecolor{deepgreen}{rgb}{0,0.5,0}

\usepackage{listings}

\newcommand\pythonstyle{\lstset{
language=Python,
basicstyle=\ttm,
otherkeywords={self},             % Add keywords here
keywordstyle=\ttb\color{deepblue},
emph={MyClass,__init__},          % Custom highlighting
emphstyle=\ttb\color{deepred},    % Custom highlighting style
stringstyle=\color{deepgreen},
frame=tb,                         % Any extra options here
showstringspaces=false            % 
}}

\lstnewenvironment{python}[1][]
{
\pythonstyle
\lstset{#1}
}
{}

\newcommand\pythoninline[1]{{\pythonstyle\lstinline!#1!}}

\newtheorem{theorem}{Theorem}

\newcommand{\algoname}[0]{Eiffel}

\begin{document}
%\tracingall

\title{\Large \bf \algoname{}: Efficient and Flexible Software Packet Scheduling}

\author{ {\rm Ahmed Saeed$^\dagger$, Yimeng Zhao$^\dagger$,  
Nandita Dukkipati$^\ast$, 
Mostafa Ammar$^\dagger$, Ellen Zegura$^\dagger$,} \\ 
{\rm Khaled Harras$^\ddagger$, Amin Vahdat$^\ast$}\\
$^\dagger$Georgia Institute of Technology,
$^\ast$Google,
$^\ddagger$Carnegie Mellon University\\
{\normalfont\textbf{\normalsize  This paper is an extended version of the paper on Eiffel that will be published in USENIX NSDI 2019.}}
}
%\author{Paper \#1010101}
%\date{\vspace{-0.3in}}
%\author{}

\maketitle
\thispagestyle{empty}

%\begin{abstract}
\subsection*{Abstract}
Packet scheduling determines the ordering of packets in a queuing data structure with respect to some ranking function that is mandated by a scheduling policy. It is the core component in many recent innovations to optimize network performance and utilization. 
%Packet scheduling is used for network resource allocation, meeting network-wide delay objectives, or providing isolation and differentiation of service. 
Our focus in this paper is on the design and deployment of packet scheduling in software. Software schedulers have several advantages over hardware including shorter development cycle and flexibility in functionality and deployment location. 
We substantially improve current software packet scheduling performance, while maintaining flexibility, by exploiting underlying features of packet ranking; namely, packet ranks are integers and, at any point in time, fall within a limited range of values. 
%This allows us to rely on integer priority queues, compared to existing ranking algorithms, that rely on comparison-based priority queues that assume continuous ranks with infinite range. 
We introduce Eiffel, a novel programmable packet scheduling system. At the core of Eiffel is an integer priority queue based on the Find First Set (FFS) instruction and designed to support a wide range of policies and ranking functions efficiently. As an even more efficient alternative, we also propose a new approximate priority queue that can outperform FFS-based queues for some scenarios. To support flexibility, Eiffel introduces novel programming abstractions to express scheduling policies that cannot be captured by current, state-of-the-art scheduler programming models. We evaluate Eiffel in a variety of settings and in both kernel and userspace deployments.  We show that it outperforms state of the art systems by 3-40x in terms of either number of cores utilized for network processing or number of flows given fixed processing capacity.

\vspace{-0.1in}
\section{Introduction}
\label{sec:intro}
\vspace{-0.1in}
Packet scheduling is the core component in many recent innovations to optimize network performance and utilization. Typically, packet scheduling targets network-wide objectives (e.g., meeting strict deadlines of flows \cite{pdq}, reducing flow completion time \cite{pfabric}), or provides isolation and differentiation of service (e.g., through bandwidth allocation \cite{bwe, swan} or Type of Service levels \cite{PASE,PIAS,qjump}).  It is also used for resource allocation within the packet processing system (e.g., fair CPU utilization in middleboxes \cite{Wang:2014:FTP:2674005.2675010,DRFQ} and software switches \cite{softnic}).

Packet scheduling determines the ordering of packets in a {\em queuing data structure} with respect to some {\em ranking function} that is mandated by a {\em scheduling policy}. In particular, as packets arrive at the scheduler they are {\em enqueued}, a process that involves ranking based on the scheduling policy and ordering the packets according to the rank. Then, periodically, packets are {\em dequeued} according to the packet ordering. In general, the dequeuing of a packet might, for some scheduling policies, prompt recalculation of ranks and a reordering of the remaining packets in the queue. A packet scheduler should be {\em efficient} by performing a minimal number of operations on packet enqueue and dequeue thus enabling the handling of packets at high rates. It should also be {\em flexible} by providing the necessary abstractions to implement as many scheduling policies as possible. %Schedulers can be {\em work-conserving}, meaning a packets will be dequeued back-to-back as long as there scheduler is occupied, or {\em non-work-conserving}, where packets are dequeued according to some inter-departure timing constraints dictated by traffic pacing or shaping requirements. 

In modern networks, hardware and software both play an important role \cite{211243}. While hardware implementation of network functionality will always be faster than its corresponding software implementation, software schedulers have several advantages. First, the short development cycle and flexibility of software makes it an attractive replacement or precursor for hardware schedulers.
Second, the number of rate limiters and queues deployed in hardware implementations typically lags behind network needs.  For instance, three years ago, network needs were estimated to be in the tens of thousands of rate limiters \cite{SENIC} while hardware network cards offered 10-128 queues \cite{intel_old1}. 
Third, software packet schedulers can be deployed in multiple platforms and locations, including middleboxes as Virtual Network Functions and end hosts (e.g., implementation based on BESS \cite{softnic}, or OpenVSwitch \cite{openvswitch}).
Hence, we assert that software solutions will always be needed to replace or augment hardware schedulers \cite{billaud2013hclock,hubert2002linux,carousel,QFQ, kogan2017programmable}. However, as will be discussed in  Section \ref{sec:background_and_objectives}, current software schedulers do not meet our efficiency and flexibility objectives. 

%For example, \cite{kogan2017programmable}, not our flexibility objective \cite{QFQ}, and in some cases fail to satisfy either objective \cite{billaud2013hclock, hubert2002linux, carousel}.  

%The flexibility and efficiency requirements for software scheduler is motivated by requirements of modern network operators. 
%We can already see these requirements in two trends in research and development of other network components. First, network operators prefer programmable network components which have shorter development cycles. The programmability trend gave rise to industry standards like OpenFlow \cite{mckeown2008openflow} and P4 \cite{Bosshart:2014:PPP:2656877.2656890}. 

{\em Our focus in this paper is on the design and implementation of efficient and flexible packet scheduling in software.} The need for programmable schedulers is rising as more sophisticated policies are required of networks \cite{feamster2017and, pifo} with schedulers deployed at multiple points on a packet's path.  It has proven difficult to achieve scheduler efficiency in software schedulers, especially handling packets at high line rates, without limiting the supported scheduling policies \cite{carousel,pifo,hubert2002linux, carousel,billaud2013hclock, QFQ}. 
%In particular, most deployed software schedulers are optimized by constraining them to support a limited set of functions (e.g., rate limiting \cite{hubert2002linux, carousel} or weighted fair queuing \cite{billaud2013hclock, QFQ}). 
Furthermore, CPU-efficient implementation of even the simplest scheduling policies is still an open problem for most platforms. For instance, kernel packet pacing can cost CPU utilization of up to 10\% \cite{carousel} and up to 12\% for hierarchical weighted fair queuing scheduling in NetIOC of VMware's hypervisor \cite{vmware_2}. This overhead will only grow as more programmability is added to the scheduler, assuming basic building blocks remain the same (e.g., OpenQueue \cite{kogan2017programmable}). {\color{black} The inefficiency of the discussed systems stems from relying on $O(\log n)$ comparison-based priority queues.
%Hence, taming this overhead requires a reexamination of building blocks for flexible packet schedulers. 

At a fundamental level, a scheduling policy that has $m$ ranking functions associated with a packet (e.g., pacing rate, policy-based rate limit, weight-based share, and deadline-based ordering) typically requires $m$ priority queues in which this packet needs to be enqueued and dequeued \cite{pifo-hotnets}, which translates roughly to $O(m\log n)$ operations per packet for a scheduler with $n$ packets enqueued. We show how to reduce this overhead to $O(m)$ for any scheduling policy (i.e., constant overhead per ranking function).}

Our approach to providing both flexibility and efficiency in software packet schedulers is two fold.  
%First, we reexamine packet priority queuing and find that existing systems rely on comparison-based priority queues with a worst case performance of $O(\log n)$. We find that packet ranks are integers that have predetermined ranges and that many packets will typically have equal rank (\S\ref{sec:background_and_objectives}). We leverage this insight 
First, we observe (\S\ref{sec:background_and_objectives}) that packet ranks can be represented as integers that at any point in time fall within a limited window of values.
We exploit this property (\S\ref{sec:ffs}) to employ integer priority queues that have $O(1)$ overhead for packet insertion and extraction. We achieve this by proposing a modification to priority queues based on the Find First Set (FFS) instruction, found in most CPUs, to support a wide range of policies and ranking functions efficiently. We also propose a new approximate priority queue that can outperform FFS-based queues for some scenarios (\S\ref{sec:gq}). Second, we observe (\S\ref{sec:flex}) that packet scheduling programming models (i.e., PIFO \cite{pifo} and OpenQueue \cite{kogan2017programmable}) do not support per-flow packet scheduling nor do  they support reordering of packets on a dequeue operation. 
%For example, a common scenario is where the transmission time of all queued packets belonging to a flow needs to change when changing the flow rank such as the case and such rank can change on packet enqueue or dequeue (e.g., Longest Queue First (LQF) where both incoming and outgoing packets change a flow's rank). 
We augment the PIFO scheduler programming model to capture these two abstractions.

We introduce \algoname{}, an efficient and flexible software scheduler that instantiates our proposed approach. Eiffel is a software packet scheduler that can be deployed on end-hosts and software switches to implement any scheduling algorithm. To demonstrate this we implement Eiffel (\S\ref{sec:impl}) in: 1) the kernel as a Queuing Discipline (qdisc) and compare it to Carousel \cite{carousel} and FQ/Pacing \cite{fq} and 2) the Berkeley Extensible Software Switch (BESS) \cite{bess,softnic} using Eiffel-based implementations of pFabric \cite{pfabric} and hClock \cite{billaud2013hclock}. We evaluate Eiffel in both settings (\S\ref{sec:eval}). Eiffel outperforms Carousel by 3x and FQ/Pacing by 14x in terms of CPU overhead when deployed on Amazon EC2 machines with line rate of 20 Gbps.  We also find that an Eiffel-based implementation of pFabric and hClock outperforms an implementation using comparison-based priority queues by 5x and 40x respectively in terms of maximum number of flows given fixed processing capacity and target rate.

%The rest of this paper is structured as follows. Section~\ref{sec:background_and_objectives} discusses the design space of packet schedulers and introduces our objectives. Section~\ref{sec:eiffel} and \ref{sec:impl} introduce our proposed system and its implementation, respectively. Section \ref{sec:eval} presents our evaluation. We discuss some related work in Section \ref{sec:related} and conclude the paper in Section \ref{sec:conc}. 

\vspace{-0.1in}

\section{Background and Objectives}
\label{sec:background_and_objectives}
\vspace{-0.1in}

\begin{table*}[t!]
\footnotesize
\centering
\begin{tabular}{|c|l|c|c|c|c|c|c|} 
\hline
\multirow{3}*{System} & \multirow{3}*{Efficiency} & \multirow{3}*{HW/SW}  & \multicolumn{4}{c|}{Flexibility} & \multirow{3}*{Notes} \\
\cline{4-7}
 & & & Unit of &  Work- &   Supports    & \multirow{2}*{Programmable} & \\
&  & & Scheduling & Conserving & Shaping  &     &  \\

%\multirow{2}*{System} & \multirow{2}*{Overhead} & \multirow{2}*{HW/SW} & Work & Packet     & Flow       & \multirow{2}*{Shaping}  &    \multirow{2}*{Notes} \\
%     &  & & Conserving & Scheduling & Scheduling &     &  \\
 \hline

FQ/Pacing qdisc \cite{fq} & $O(\log n)$ & SW & Flows & No & \textbf{Yes} & No & Only non-work conserving FQ \\

%{\color{Mycolor2}\textbf{No}} & {\color{Mycolor2} \textbf{No}} & {\color{Mycolor2}\textbf{Yes}} & {\color{Mycolor1}\textbf{Yes}}  &  - \\

%FQ/Pacing \cite{fq} & $O(\log n)$ & SW & {\color{black}{No}} & {\color{black} No} & {\color{black}{No}} & {\color{black}\textbf{Yes}}  & Only FQ Sched. \\

%UPS \cite{UPS} & $O(1)$ & HW+SW & {\color{black}\textbf{Yes}} & {\color{black}\textbf{Yes}} & {\color{black}\textbf{No}} & {\color{black}\textbf{No}}   &  Req. uncongested net.\\

hClock \cite{billaud2013hclock} & $O(\log n)$ & SW & Flows & {\color{black}\textbf{Yes}}  & {\color{black}\textbf{Yes}}  & No &   Only HWPQ Sched. \\

Carousel \cite{carousel} & $O(1)$ & SW & Packets & {\color{black}{No}}  & {\color{black}\textbf{Yes}} & No & Only non-work conserving sched.  \\%& Hierarchical shaping policies \\

%QFQ \cite{QFQ} & $O(1)$ & SW & \textbf{Yes} & No & \textbf{Yes} & No & Only WFQ \\ 

OpenQueue \cite{kogan2017programmable} & $O(\log n)$ & SW & Packets \& Flows & {\color{black}\textbf{Yes}} & No & On enq/deq    & Inefficient building blocks \\

PIFO \cite{pifo}   & $O(1) $ & HW & Packets & {\color{black}\textbf{Yes}}  &  {\color{black}\textbf{Yes}} & On enq &  Max. \# flows 2048\\

\textbf{\textit{Eiffel}} & $O(1)$ & SW & Packets \& Flows & {\color{black}\textbf{Yes}} &  {\color{black}\textbf{Yes}} & On enq/deq &  - \\%& \hspace{1.4in}-\\
\hline
\end{tabular}
\vspace{-0.1in}
\caption{Proposed work in the context of the state of the art in scheduling}
\vspace{-0.1in}
\label{table:1}
\end{table*}

In modern networks, packet scheduling can easily become the system bottleneck. This is because schedulers
are burdened with the overhead of maintaining a large number of buffered packets sorted according to scheduling policies. {\color{black}Despite the growing capacity of modern CPUs, packet processing overhead remains a concern. Dedicating CPU power to networking takes from CPU capacity that can be dedicated to VM customers especially in cloud settings \cite{211249}.}
One approach to address this overhead is to optimize the scheduler for a specific scheduling policy \cite{fq,linuxhtb,billaud2013hclock, carousel, QFQ}. However, with specialization two problems linger. First, in most cases inefficiencies remain because of the typical reliance on generic default priority queues in modern libraries (e.g., RB-trees in kernel and Binary Heaps in C++). Second, even if efficiency is achieved, through the use of highly efficient specialized data structures (e.g., Carousel \cite{carousel} and QFQ \cite{QFQ}) or hybrid hardware/software systems (e.g. SENIC \cite{SENIC}), this efficiency is achieved at the expense of programmability. The Eiffel system we develop in this paper is designed to be both efficient and programmable.  In this section we examine these two objectives, show how existing solutions fall short of achieving them and highlight our approach to successfully combine efficiency with flexibility.
%are still problematic as relying on a small library of scheduling policies that can operate at line rate is not compatible with network programmability. A programmable network should ideally allow network operators to specify any arbitrary scheduling policy while maintaining efficiency. %Current approaches to data plane programming in software has worked on improving processing capacity in terms of packets per second for the system as a whole. We show that the scheduling or queuing policy significantly affect system capacity. 
%For the rest of this section, we examine these two problems in more detail and derive from them the two design objectives of the Eiffel system. %Then, we introduce Eiffel, a system that leverages our new data structures and programmability constructs to implement a highly efficient and flexible packet scheduling in software.

\textbf{Efficient Priority Queuing:} Priority queuing is fundamental to computer science with a long history of theoretical results. Packet priority queues are typically developed as comparison-based priority queues \cite{fq,billaud2013hclock}. A well known result for such queues is that they require $O(\log n)$ steps for either insertion or extraction for a priority queue holding $n$ elements \cite{Thorup2006Equivalence}. This applies to data structures that are widely used in software packet schedulers such as RB-trees, used in kernel Queuing Disciplines, and Binary Heaps, the standard priority queue implementation in C++.  
%Approaching packet priority queues for packet ranking through comparison-based algorithms leads to a worst case overhead of $O(\log n)$ per packet, where $n$ is the number of packets in the queue. 

Packet queues, however, have the following characteristics that can be exploited  to significantly lower the overhead of packet insertion and extraction:
\vspace{-0.1in}
\begin{itemize}[leftmargin=*]
\item \textit{Integer packet ranks:} Whether it is deadlines, transmission time, slack time, or priority, the calculated rank of a packet can always be represented as an integer. 
\vspace{-0.1in}
\item \textit{Packet ranks have specific ranges:}  At any point in time, the ranks of packets in a queue will typically fall within a limited range of values {\color{black}(i.e., with well known maximum and minimum values)}. This range is policy and load dependent and can be determined in advance by operators (e.g., transmission time where packets can be scheduled a maximum of a few seconds ahead, flow size, or known ranges of strict priority values). {\color{black}Ranges of priority values are diverse, ranging from just eight levels \cite{1637340}, to 50k for a queue implementing per flow weighted fairness which requires a number of priorities corresponding to the number of flows (i.e., 50k flows on a video server \cite{carousel}), and up to 1 million priorities for a time indexed priority queue \cite{carousel}. }
%can be based on \textit{time}, with well known time scale and range (i.e., hundreds of microseconds to a few milliseconds scale with a range of maximum a few seconds), or \textit{rank} with a range that is policy and load dependent and known in advance by operators (e.g., typical flow size or a known range priority values).
\vspace{-0.1in}
\item \textit{Large numbers of packets share the same rank:} Modern line rates are in the ranges of 10s to 100s of Gbps. Hence, multiple packets are bound to be transmitted with nanosecond time gaps. This means that packets with small differences in their ranks can be grouped and said to have the same rank with minimal or no effect on the accurate implementation of the scheduling policy. {\color{black}For instance, consider a busy-polling-based packet pacer that can dequeue packet at fixed intervals (e.g., order of 10s of nanoseconds). In that scenario, packets with gaps smaller than 10 nanoseconds can be considered to have the same rank.}

%that while their absolute rank might have slight differences,
%packets transmitted at the same time practically have the same rank. This makes the ranks amenable to grouping or approximation. 
\end{itemize}
\vspace{-0.1in}

These characteristics make the design of a packet priority queue effectively the design of bucketed integer priority queues over a finite range of rank values $[0, C]$ with number of buckets $N$, each covering $C/N$ interval of the range. 
%The range of the queue, $C$, corresponds to the range of the rank values.
The number of buckets, and consequently the range covered by each bucket, depend on the required ranking granularity which is a characteristic of the scheduling policy. The number of buckets is typically in the range of a few thousands to hundreds of thousands. Elements falling within a range of a bucket are ordered in FIFO fashion. 
%which correpsonds to 
%the number of buckets.  which is the number of priority values generated by the ranking function.
Theoretical complexity results for such bucketed integer priority queues are reported in \cite{van1975preserving,fredman1990blasting,Thorup2006Equivalence}. 

Integer priority queues do not come for free. Efficient implementation of integer priority queues requires pre-allocation of buckets and meta data to access those buckets. 
%The number of buckets depends on the scheduling policy 
In a packet scheduling setting the number of buckets is fixed, making the overhead per packet a constant whose value is logarithmic in the number of buckets{\color{black}, because searching is performed on the bucket list not the list of elements}. 
%(e.g., four steps to retrieve the minimum element in a million buckets regardless of the number of queued packets). 
Hence, bucketed integer priority queues achieve CPU efficiency at the expense of maintaining elements unsorted within a single bucket and pre-allocation of memory for all buckets. Note that maintaining elements unsorted within a bucket is inconsequential because packets within a single bucket effectively have equivalent rank.
%the first problem is inconsequential in a packet scheduling setting due to the third characteristic mentioned above. 
Moreover, the memory required for buckets, in most cases, is minimal (e.g., tens to hundreds of kilobytes), which is consistent with earlier work on bucketed queues \cite{carousel}.
%the main overhead of bucketed integer priority queues is in terms of memory needed to preallocate all buckets,  unlike comparison-based priority queues that allocate memory per element. However, this overhead is small, especially in software, as the overhead per bucket is typically just a few bytes. 
Another advantage of bucketed integer priority queues is that elements can be (re)moved with $O(1)$ overhead. This operation is used heavily in several scheduling algorithms (e.g., hClock \cite{billaud2013hclock} and pFabric \cite{pfabric}). 

Recently, there has been some attempts to employ data structures specifically developed or re-purposed for efficiently implementing specific packet scheduling algorithms. For instance, Carousel \cite{carousel}, a system developed for rate limiting at scale, relies on Timing Wheel \cite{tw}, a data structure that can support time-based operations in $O(1)$ and requires comparable memory to our proposed approach. However, Timing Wheel supports only non-work conserving time-based schedules in $O(1)$. {\color{black}Timing Wheel is efficient as buckets are indexed based on time and elements are accessed when their deadline arrives. However, Timing Wheel does not support operations needed for non-work conserving schedules (i.e., \verb|ExtractMin| or \verb|ExtractMax|).}
Another example is efficient approximation of popular scheduling policies (e.g., Start-Time Fair Queueing \cite{Goyal:1996:SFQ:248156.248171} 
as an approximation of Weighted Fair Queuing \cite{Demers:1989:ASF:75246.75248}, or the more recent Quick Fair Queuing (QFQ) \cite{QFQ}). This approach of developing a new system or a new data structure per scheduling policy does not provide a %clear, or even vague, 
path to the efficient implementation of more complex policies. Furthermore, it does not allow for a truly programmable network. These limitations lead us to our first objective for Eiffel:

%Recently, there has been some attempts to improve these issues by employing data structures specifically developed or re-purposed for efficiently implementing specific packet scheduling algorithms. For instance, Carousel \cite{carousel}, a system developed for rate limiting at scale, relies on Timing Wheel \cite{tw}, a data structure that can support time-based operations in $O(1)$. However, Timing Wheel supports only non-work conserving time-based schedules in $O(1)$. 

\noindent\textit{Objective 1: Develop data structures that can be employed for any scheduling algorithm providing $O(1)$ processing overhead per packet leveraging integer priority queues (\S\ref{sec:pq}).}

\textbf{Flexibility of Programmable Packet Schedulers:} There has been recent interest in developing flexible, programmable, packet schedulers \cite{pifo, kogan2017programmable}. This line of work is motivated by the support for programmability in all aspects of modern networks. 
%These programmable schedulers should support the diverse requirements and specifications of scheduling policies. 
{\color{black} Specifically, several packet scheduling algorithms have been proposed each aiming at achieving a specific and diverse set of objectives. Some of such policies include per flow rate limiting \cite{carousel}, hierarchical rate limiting with strict or relative prioritization \cite{billaud2013hclock,linuxhtb}, optimization of flow completion time \cite{pfabric}, and joint network and CPU scheduling \cite{DRFQ}, several of which are used in practice.}

{\color{black}Work on programmable schedulers focuses on providing the infrastructure for network operators to define their own scheduling policies. This approach improves on the current standard approach of providing a small fixed set of scheduling policies as currently provided in modern switches. A programmable scheduler provides building blocks for customizing packet ranking and transmission timing. Proposed programmable schedulers differ based on the flexibility of their building blocks. A flexible scheduler allows a network operator to specify policies according to the following specifications:}
\vspace{-0.1in}
\begin{itemize}[leftmargin=*]
\item \textit{Unit of Scheduling:} Scheduling policies operate either on per packet basis (e.g., pacing) or on per flow basis (e.g., fair queuing). This requires a model that provides abstractions for both.
\vspace{-0.1in}
\item \textit{Work Conservation:} Scheduling policies can also be work-conserving or non-work-conserving. 
\vspace{-0.1in}
\item \textit{Ranking Trigger:} Efficient implementation of policies can require  ranking packets on their enqueue, dequeue, or both. 
\end{itemize}
\vspace{-0.1in}

%For instance, scheduling policies can operate on different units of scheduling (e.g., per packet operations for pacing or per flow operations for fair queuing). Scheduling policies can also be work-conserving or non-work-conserving. 
%Control plane programming has been enabled by advancements in OpenFlow \cite{mckeown2008openflow} and recently data plane programming is being pushing forward through advancements in P4 \cite{Bosshart:2014:PPP:2656877.2656890}. 

Recent programmable packet schedulers export primitives that enable the specification of a scheduling policy and its parameters, often within limits.
%This flexibility, however, can detract from efficiency. 
The PIFO scheduler programming model is the most prominent example \cite{pifo}. It is implemented in hardware relying on Push-In-First-Out (PIFO) building blocks where packets are ranked only on enqueue. The scheduler is programmed by arranging the blocks to implement different scheduling policies. Due to its hardware implementation, the PIFO model employs compact constructs with considerable flexibility. However, PIFO remains very limited in its capacity (i.e., PIFO can handle  a a maximum of 2048 flows at line rate), and expressiveness (i.e., PIFO can't express per flow scheduling). OpenQueue is an example of a flexible programmable packet scheduler in software \cite{kogan2017programmable}. However, the flexibility of OpenQueue comes at the expense of having three of its building blocks as priority queues, namely queues, buffers, and ports.  This overhead, even in the presence of efficient priority queues, will form a memory and processing overhead. Furthermore, OpenQueue does not support non-work-conserving schedules.

The design of a flexible and efficient packet scheduler remains an open research challenge. {\color{black} It is important to note here that the efficiency of programmable schedulers is different from the efficiency of policies that they implement. An efficient programmable platform aims to reduce the overhead of its building blocks (i.e., Objective 1) which makes the overhead primarily a function of the complexity of the policy itself. Thus, the efficiency of a scheduling policy becomes a function of only the number of building blocks required to implement it. Furthermore, an efficient programmable platform should allow the operator to choose policies based on their requirements and available resources by allowing the platform to capture a wide variety of policies.} To address this challenge, we choose to extend the PIFO model due to its existing efficient building blocks. In particular, we introduce flows as a unit of scheduling in the PIFO model. We also allow modifications to packet ranking and relative ordering both on enqueue and dequeue. 

%However, programmable schedulers remain very limited either in their capacity (i.e., PIFO model has limitation of scheduling a maximum of 2048 flows at line rate \cite{pifo}), performance (i.e., OpenQueue results in performance degradation \cite{kogan2017programmable}), and expressiveness (i.e., PIFO and OpenQueue can't express per flow scheduling and OpenQueue doesn't handle non-work-conserving schedules). %Eiffel resolves all three problems by improving performance and capacity through relying on efficient priority queues. Eiffel also improves expressiveness by introducing new programming abstractions which extend the PIFO programming model. 
%However, OpenQueue suffers from the same expressiveness limitations as the PIFO model and it does not support shaping altogether. Furthermore, it relies on $O(\log n)$ data structures making its performance worse than other existing tailored software schedulers and limiting its attractiveness for any practical application. 
%Ultimately, the design of truly flexible and efficient software schedulers remains an open research challenge, which leads us to our second objective:

\noindent\textit{Objective 2: Provide a fully expressive scheduler programming abstraction by extending the PIFO model (\S\ref{sec:flex}).}

\textbf{Eiffel's place in Scheduling Research Landscape:}
This section reviewed scheduling support in software
\footnote{Scheduling is widely supported in hardware switches using a short list of scheduling policies, including shaping, strict priority, and Weighted Round Robin \cite{arista_1,arista_2,cisco_1,pifo}. An approach to efficient hardware packet scheduling relies on pipelined-heaps  \cite{bhagwan2000fast,4154755,wang2013per} to help position Eiffel. Pipelined-heaps are composed of piplined-stages for enqueuing and dequeuing elements in a priority queue. However, such approaches are not immediately applicable to software.}.
%To better understand the design space of the Eiffel system, 
Table~\ref{table:1} summarizes the discussed related work. Eiffel fills the gap in earlier work by being the first efficient $O(1)$ and programmable software scheduler. It can support both per flow policies (e.g., hClock and pFabric) and per packet scheduling policies (e.g., Carousel). It can also support both work-conserving and non-work-conserving schedules. %  because it relies on hardware optimization that allow for pipelining operations of the queue in hardware.

%and shows how the Eiffel, as efficient and highly programmable schedulers that can operate at scale, fills the gap by satisfying the two objectives we discuss above. In particular, none of the earlier work has addressed implementing flow scheduling efficiently. Furthermore, programmable schedulers suffer from scalability issues. 

%\vspace{-0.2in}

\section{Eiffel Design}
\label{sec:eiffel}
\vspace{-0.1in}

Figure~\ref{fig:arch} shows the architecture of \algoname{} with four main components: 1) a packet annotator to set the input to the enqueue component (e.g., packet priority), 2) an enqueue component that calculates a rank for incoming packets, 3) a queue that holds packets sorted based on their rank, and 4) a dequeue component which is triggered to re-rank elements in the queue, for some scheduling algorithms. Eiffel leverages and extends the PIFO scheduler programming model to describe scheduling policies \cite{pifo,pifo_github}. %Scheduler controller uses configuration files to program the scheduler with specific policies. 
The functions of the packet annotator, the enqueue module, and the dequeue module are derived in a straightforward manner from the scheduling policy. 
The only complexity in the Scheduler Controller, namely converting scheduling policy description to code, has been addressed in earlier work on the PIFO model \cite{pifo_github}.  The two complicated components in this architecture, therefore, correspond with the two objectives discussed in the previous section: the Queue \textit{(Objective 1)} and The Scheduling policy Description \textit{(Objective 2)}. For the rest of this section, we explain our efficient queue data structures along with our extensions to the programming model used to configure the scheduler.

\vspace{-0.1in}
\subsection{Priority Queueing in Eiffel}
\label{sec:pq}
\vspace{-0.08in}
\begin{figure}[t]
\centering
\includegraphics[width=0.5\textwidth]{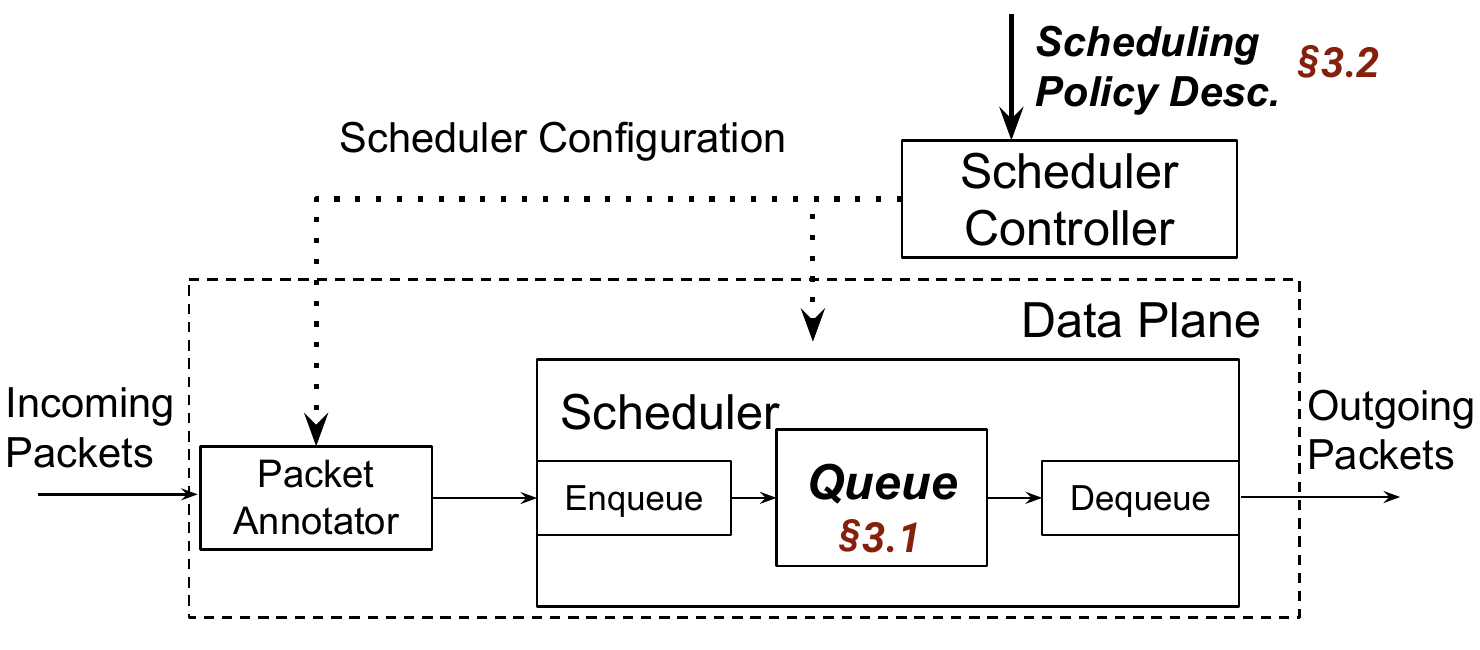}
\vspace{-0.1in}
\caption{Eiffel programmable scheduler architecture highlighting Eiffel's extensions.}
\label{fig:arch}
\vspace{-0.1in}
\end{figure}

A priority queue maintains a list of elements, each tagged with a priority value. A priority queue supports one of two operations efficiently: \verb|ExtractMin| or \verb|ExtractMax| to get the element with minimum or maximum priority respectively.
Our goal, as stated in Objective 1 in the previous section, is to enable these operations with $O(1)$ overhead. 
To this end we first develop a circular extension of efficient priority queues that rely on the FindFirstSet (FFS) operation, found in all modern CPUs \cite{intel_ia, amd_ia}. Our extensions allow FFS-based queues to operate over large moving ranges while maintaining CPU efficiency. We then improve on the FFS-based priority queue by introducing the approximate gradient queue, which can perform priority queuing in $O(1)$ under some conditions. The approximate priority queue can outperform the FFS-based queue by up to 9\% for scenarios of a highly occupied bucketed priority queue (\S\ref{sec:eval_2}). {\color{black}Note that for all Integer Priority Queues discussed in this section, enqueue operation is trivial as buckets are identified by the priority value of their elements. This makes the enqueue operation a simple bucket lookup based on the priority value of the enqueued element.}

 %It is trivial to convert a min priority (i.e., a priority queue that supports \verb|ExtractMin|)  queue to a max priority queue and vice versa. 
%, where \verb|ExtractMin| is more prevalent in packet scheduling algorithms. 

%In this section, we use examples with either \verb|ExtractMin| or \verb|ExtractMax|, which correspond to min priority queues or max priority queues, to facilitate explanation. However, it is important to note that it is trivial to convert a min priority queue to a max priority queue and vice versa. 

%Priority queuing is a very expensive operation whose efficiency depends on the rankings of enqueued elements. 
%We introduce two data structures: 1) Cicular FindFirstSet (FFS)-based bucketed integer priority queue whose overhead is a function of number of buckets, and 2) Approximate Gradient queue  which has a constant overhead under some conditions. We show a comprehensive comparison between the two data structures in Section \ref{sec:eval_2}.
\vspace{-0.1in}

\subsubsection{Circular FFS-based Queue (cFFS)}
\label{sec:ffs}

\begin{figure}
    \centering
    \begin{minipage}{0.22\textwidth}
%\begin{figure}
        \centering
        \includegraphics[width=0.9\textwidth]{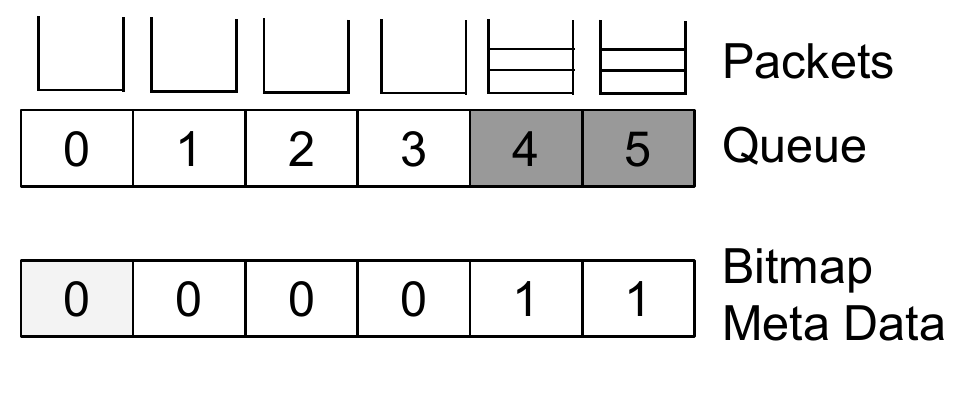} % first figure itself
        \caption{FFS-based queue where FFS of a bit-map of six bits can be processed in O(1).}
        \label{fig:basic_ffs}
    \end{minipage}\hfill
    \begin{minipage}{0.22\textwidth}
%\end{figure}
%\begin{figure}
        \centering
        \includegraphics[width=0.9\textwidth]{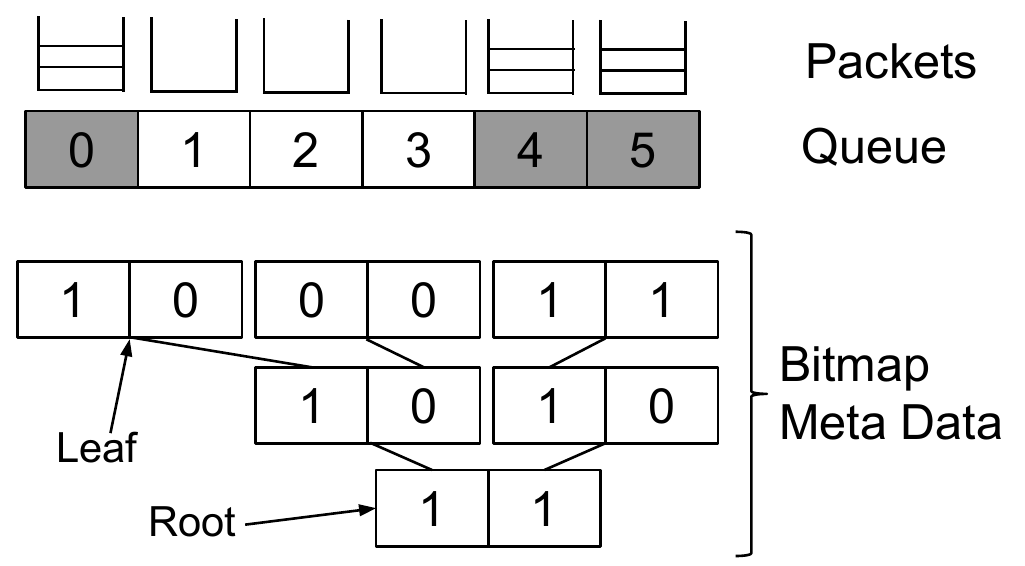} % second figure itself
        \caption{Hierarchical FFS-based queue where FFS of bit-map of two bits can be processed in O(1) using a 3-level hierarchy}
        \label{fig:hie_ffs}
    \end{minipage} 
    \vspace{-0.1in}
\end{figure}

%The difficulty of implementing data structures with theoretical tight bounds prompted work on data structures with looser bounds but that are easy and efficient to implement. For instance,  

FFS-based queues are bucketed priority queues with a bitmap representation of queue occupancy. Zero represents an empty bucket, and one represents a non-empty bucket. FFS produces the index of the leftmost set bit in a machine word in constant time. {\color{black}All modern CPUs support a version of Find First Set at a very low overhead (e.g.,  Bit-Scan-Forward (\verb|BSR|) takes three cycles to complete \cite{intel_ia}).} Hence, a priority queue, with a number of buckets equal to or smaller than the width of the word supported by the FFS operation can obtain the smallest set bit, and hence the element with the smallest priority, in $O(1)$ (e.g., Figure~\ref{fig:basic_ffs}). In the case that a queue has more buckets than the width of the word supported by a single FFS operation, a set of words can be processed sequentially to represent the queue with every bit representing a bucket. This results in an $O(M)$ algorithm that is very efficient for very small $M$, where $M$ is the number of words. For instance, realtime process scheduling in the linux kernel has a hundred priority levels. An FFS-based priority queue is used where FFS is applied sequentially on two words, in case of 64-bit words, or four words in case of 32-bit words \cite{linux_rt}. {\color{black} This algorithm is not efficient for large values of $M$ as it requires scanning all words, in the worst case, to find the index of the highest priority element.} FFS instruction is also used in QFQ to sort groups of flows based on the eligibility for transmission where the number of groups is limited to a number smaller than 64 \cite{QFQ}. {\color{black}QFQ is an efficient implementation of fair queuing which uses FFS efficiently over a small number of elements. However, QFQ does not provide any clear direction towards implementing other policies efficiently.}

%CPU scheduling in linux has only 140 priority levels  \cite{love2010linux}). 

To handle an even larger numbers of priority levels, hierarchical bitmaps may be used. One example is Priority Index Queue (PIQ) \cite{wang2013per}, a hardware implementation of FFS-based queues, which introduces a hierarchical structure where each node represents the occupancy of its children, and the children of leaf nodes are buckets. The minimum element can be found by recursively navigating the tree using FFS operation (e.g., Figure~\ref{fig:hie_ffs} for a word width of two). Hierarchical FFS-based queues have an overhead of $O(\log_{w}N)$ where $w$ is the width of the word that FFS can process in $O(1)$ and $N$ is the number of buckets. It is important to realize that, for a given scheduling policy, the value of $N$ is a given fixed value that doesn't change once the scheduling policy is configured. Hence, a specific instance of a Hierarchical FFS-based queue has a constant overhead independent of the number of enqueued elements. {\color{black}In other words, once an implementation is created $N$ does not change.}
%However, for a large $N$, a more complex algorithm is needed to reduce the number of operations needed to find the smallest element. 
%Priority Index Queue (PIQ) \cite{wang2013per}, a hardware implementation of FFS-based queues, introduces a hierarchical structure where each level represents the occupancy of multiple buckets, and the minimum element can be found by recursively navigating the tree using FFS operation (e.g., Figure~\ref{fig:hie_ffs} for a word width of two).
%This algorithm results in an overhead of $O(\log_{w}N)$ where $w$ is the width of the word that FFS can process in $O(1)$. 

Hierarchical FFS-based queues only work for a fixed range of priority values. However, as discussed earlier, typical priority values for packets span a moving range. PIQ avoids this problem by assuming support for the universe of possible values of priorities. This is an inefficient approach because it requires generating and maintaining a large number buckets, with relatively few of them in use at any given time. %as at any point in time priority values usually lie within a specific range that is moving over time.
%priority values, which still requires handling wraparound of priority values but such case is not frequent. However, this approach requires allocating millions of buckets to cover such a wide range of values. 

Typical approaches to operating over a large moving range while maintaining a small memory footprint rely on {\em circular queues}. Such queues rely on the $mod$ operation to map the moving range to a smaller range. However, the typical approach to circular queuing does not work in this case as it results in an incorrect bitmap. For example, if we add a packet with priority value six to the queue in Figure~\ref{fig:basic_ffs} selecting the bucket with a $mod$ operation, the packet will be added in slot zero and consequently mark the bit map at slot zero.
%To appreciate the problem, consider the example in Figure~\ref{fig:basic_ffs}. If we need to add a packet with priority value six, a circular queue will add the packet in slot zero and consequently mark the bit map at slot zero. This clearly causes incorrect behavior of the queue because the FFS operation will return the index of the first non-zero bit in the bit map which clearly is not the minimum element in this scenario. 
{Hence, once the range of an FFS-based queue is set, all elements enqueued in that range have to be dequeued before the queue can be assigned a new range {\color{black}so as to avoid unnecessary resetting of elements. 
%Because the bitmap meta data has to be reset in case any changes are made to the range of the queue. 
In that scenario, enqueued elements that are out of range are enqueued at the last bucket, and thus losing their proper ordering. Otherwise, the bitmap meta data will have to be reset in case any changes are made to the range of the queue.}}

\begin{figure}
    \centering
    \includegraphics[width=0.45\textwidth]{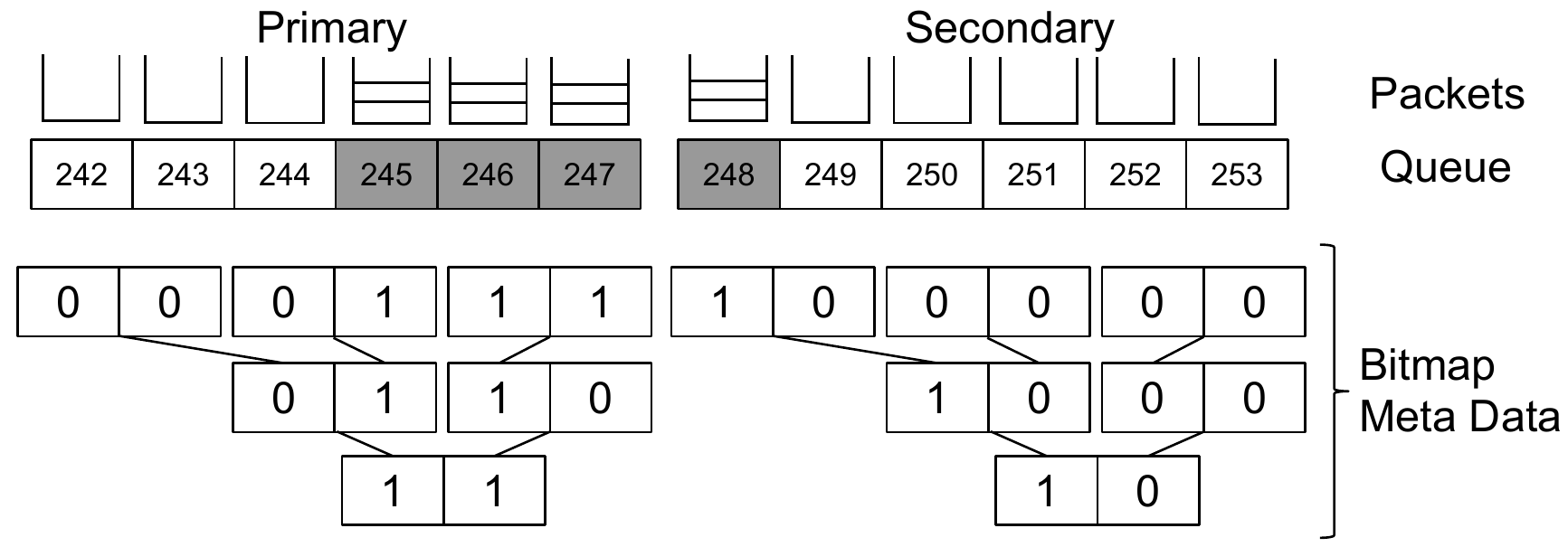} % second figure itself
    \vspace{-0.1in}
    \caption{Circular Hierarchical FFS-based queue is composed of two Hierarchical FFS-based queues, one acting as the main queue and the other as a buffer. }
    %to avoid inserting elements in the empty slots that represent elements smaller than the smallest element(i.e., avoid inserting packets with priority 248 in the 242 slot which will break the semantics of the queue).
      \vspace{-0.1in}
    \label{fig:circ_ffs}
\end{figure}

A natural solution to this problem is to introduce an overflow queue where packets with priority values outside the current range are stored. Once all packets in the current range are dequeued, packets from that ``secondary'' queue are inserted using the new range. However, this introduces a significant overhead as we have to go through all packets in the buffer every time the range advances. We solve this problem by making the secondary queue an FFS-based queue, covering the range that is immediately after the range of the queue (Figure~\ref{fig:circ_ffs}). Elements outside the  range of the secondary queue are enqueued at the last bucket in the secondary queue and their values are not sorted properly. However, we find that to not be a problem as ranges for the queues are typically easy to figure out given a specific scheduling policy.

A {\em Circular Hierarchical FFS-based queue}, referred to hereafter simply as a cFFS, maintains the minimum priority value supported by the primary queue (\verb|h_index|), the number of buckets (\verb|q_size|) per queue, two pointers to the two sets of buckets, and two pointers to the two sets of bitmaps. Hence, the queue ``circulates'' by switching the pointers of the two queues from the buffer range to the primary range and back based on the location of the minimum element along with their corresponding bitmaps. %To improve memory locality of the queue, we keep a single array of buckets that represents both primary and buffer and maintain separate bit map meta data for the primary and buffer queues.  A new element is inserted into the queue corresponding to its value, with out of range elements inserted at the last bucket of the buffer queue. The trigger for the switch operation is the primary queue becoming empty and the priority values going out of range. 

%Maintaining an FFS-based queue (e.g., Figure~\ref{fig:circ_ffs}) is a bit more complicated than a typical hierarchical FFS-based queue as it requires keeping track of which side is the primary and which is the buffer. It also requires figuring out where to insert a new element. This is performed by keeping track of the index of the head of the queue (i.e., the index of the smallest element).

%We select the queue to work with (i.e., primary or buffer) based on the priority value of interest (\verb|in_index|) which is determined based on the scenario. For instance, it can be the priority value of a new element inserted in the queue that requires identifying which side it should be enqueued to. It can also be the priority value of a dequeued element which will require identifying which side will require modification of its meta data. Finally, it can be (\verb|h_index|) itself to identify which side is currently the main side. Note that each side operates as a normal hierarchical FFS-based queue once the side selection algorithm is invoked. This maintains the overhead of the queue as $O(\log_{w}N)$ where $w$ is the width of the word that FFS CPU operation can process in O(1), and $N$ is the number of priority levels (buckets) supported by the queue. 

Note that work on efficient priority queues has a very long history in computer science with examples including van Emde Boas tree \cite{van1975preserving} and Fusion trees \cite{fredman1990blasting}. However, such theoretical data structures are complicated to implement and require complex operations. cFFS is highly efficient both in terms of complexity and the required bit operations. Moreover, it is relatively easy to implement.

\vspace{-0.18in}
\subsubsection{Approximate Priority Queuing }
\label{sec:gq}
 \vspace{-0.1in}

cFFS queues still require more than one step to find the minimum element. We explore a tradeoff between accuracy and efficiency by developing a gradient queue, a data structure that can find a \textit{near} minimum element in one step.
%We are interested in exploring techniques for trading off scheduler accuracy for efficiency. As a starting point, we describe and perform a preliminary analysis of a novel promising data structure which we label the Gradient Queue. This will serve to illustrate our approach for this direction of our proposed research.

\textbf{Basic Idea}
\label{sec:gq_basic_idea}
%\vspace{-0.08in}
\begin{figure}[!t]
\centering
\includegraphics[width=1\linewidth]{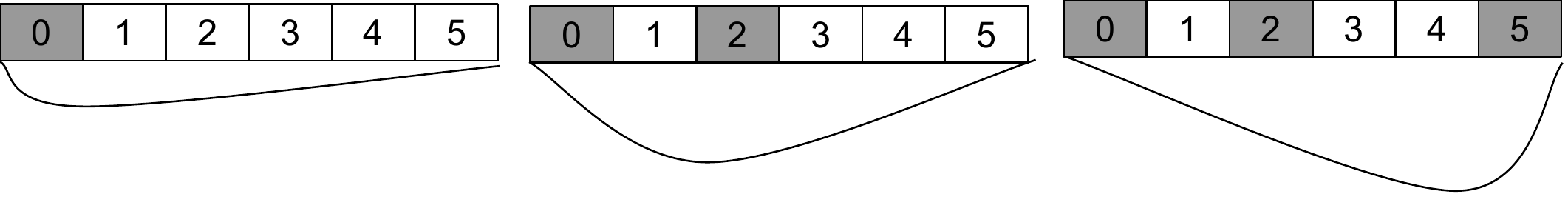}
\vspace{-0.4in}
\caption{A sketch of a curvature function for three states of a \textbf{maximum} priority queue. As the maximum index of nonempty buckets increases, the critical point shifts closer to that index.}
\vspace{-0.2in}
\label{fig:curv_fun}
\end{figure}
The Gradient Queue (GQ) relies on an algebraic approach to calculating FFS. In other words, it attempts to find the index of the most significant bit using algebraic calculations. {\em This makes it amenable to approximation}.
The intuition behind GQ is that the contribution of the most significant set bit to the value of a word is larger than the sum of the contributions of the rest of the set bits.
%Consider a bucketed priority queue. \verb|Insert| is simple as elements are inserted in the bucket corresponding to their rank. For simplicity, we explain the queue using \verb|ExtractMax|, the operation for finding the non-empty bucket with the maximum index. To understand our approach to \verb|ExtractMax|, consider the following analogy. The location of a person walking on the tightrope can be detected by finding the critical point of the curvature of the tightrope (i.e., the point where the derivative is zero). If a person and an elephant are tightrope walking, the critical point of the curvature will almost always correspond to the elephant's location. We model the occupancy of the bucket list as a curve that is ``dented'' by non-empty buckets (i.e., like a tightrope).
We consider the weight of a non-empty bucket to be proportional to its index.
Hence, Gradient Queue occupancy is represented by its curvature function. The curvature function of the queue is the sum of the weight functions of all nonempty buckets in the queue. {\color{black}More specifically, a specific curvature shape corresponds to a specific occupancy pattern. \textit{A proper weight function} ensures the uniqueness of the curvature function per occupancy pattern. It also makes finding the non-empty bucket with the maximum index equivalent to finding the critical point of the queue's curvature (i.e., the point where the derivative of the curvature function of the queue is zero).} A sample sketch of a curvature function is illustrated in Figure~\ref{fig:curv_fun}. 

%\textbf{Trade-offs:} Gradient Queue, like all bucket-sort-based queues, makes the trade-off of sacrificing memory for CPU. However, as we present later, the amount of memory needed per bucket is very small and typical networking applications do not require a large range of priorities as discussed earlier. Hence, we find that the trade-off is reasonable as the gains made in CPU efficiency are much larger than the memory cost of the data structure.

%\vspace{-0.2in}

\textbf{Exact Gradient Queue}
%\vspace{-0.06in}
On a bucket becoming nonempty, we add its weight function to the queue's curvature function, and we subtract its function when it becomes empty. We define a desirable weight function as one that is: 1) easy to differentiate to find the critical point, and 2) easy to maintain when bucket state changes between empty and non-empty. We use weight function, $2^{i} (x-i)^2$ where $i$ is the index of the bucket and $x$ is the variable in the space of the curvature function. 

This weight function results in queue curvature of the form of $ax^2-bx+c$, where the critical point is located at $x=b/2a$. Hence, we only care about $a$ and $b$ where $a=\sum_i 2^i$ and $b=\sum_i i2^i$ for all non-empty buckets $i$. The maintenance of the curvature function of the queue becomes as simple as incrementing  and decrementing $a$ and $b$ when a bucket becomes non-empty or empty respectively. Theorem~\ref{thm1}, in Appendix~\ref{app:thm}, shows that determining the highest priority non-empty queue can be calculated using $ceil (b/a)$.

A Gradient Queue with a single curvature function is limited by the the range of values $a$ and $b$ can take, which is analogous to the limitation of FFS-based queues by the size of words for which FFS can be calculated in $O(1)$. A natural solution is to develop a hierarchical Gradient Queue. This makes Gradient Queue an equivalent of FFS-based queue with more expensive operations (i.e., division is more expensive than bit operations). However, due to its algebraic nature, Gradient Queue allows for approximation that is not feasible using bit operations.

%\vspace{-0.1in}
\textbf{Approximate Gradient Queue}
\label{sec:approx}
{\color{black}Like FFS-based queues, gradient queue has a complexity of $O(\log_{w} N)$ where $w$ is the width of the representation of $a$ and $b$ and $N$ is the number of buckets. Our goal is reduce the number of steps even further for each lookup. We are particularly interested in having lookups that can be made in one operation, which can be achieved through approximation. 
The advantage of the curvature representation of the Gradient Queue compared to FSS-based approaches is that it lends itself naturally to approximation.}

A simple approximation is to make the value of $a$ and $b$ corresponding to a certain queue curvature smaller which will allow them to represent a larger number of priority values. In particular, we change the weight function to $2^{f(i)}(x-i)^{2}$ which results in
%expand the range of the gradient queue by using 
$a=\sum_i 2^{f(i)}$ and $b=\sum_i i2^{f(i)}$ where ${f(i)} = i/\alpha$ and $\alpha$ is a positive integer. This approach leads to two natural results: 1) {\color{black} the biggest gain of the approximation is that} $a$ and $b$ can now represent a much larger range of values for $i$ which eliminates the need for hierarchical Gradient Queue and allows for finding the minimum element with one step, and 2) {\color{black}the employed weight function is no longer proper. While \verb|BSR| instruction is 8-32x faster than \verb|DIV| \cite{intel_ia}, the performance gained from the reduced memory lookups required per \verb|BSR| operation.}

{\color{black}This approximation stems from using an ``improper'' weight function. This leads to breaking the two guarantees of a proper weight function, namely: 1) the curvature shape is no longer unique per queue occupancy pattern, and 2) the index of the maximum non-empty bucket no longer corresponds to the critical point of the curvature \textit{in all cases}. In other words,} the index of the maximum non-empty bucket, $M$, is no longer $ceil (b/a) $ due the fact that the weight of the maximum element no longer dominates the curvature function as the growth is sub-exponential. %Hence, the approximate queue will be most suitable in cases where the Theorem~\ref{thm1} is least violated. 
{\color{black} However, this ambiguity does not exist for all curvatures (i.e., queue occupancy patterns).}

{\color{black}We characterize the conditions under which ambiguity occurs causing error in identifying the highest priority non-empty bucket. Hence, we} identify scenarios where using the approximate queue is acceptable. The effect of $f(i)= i/\alpha$ can be described as introducing ambiguity to the value of $ceil (b/a)$. This is because exponential growth in $a$ and $b$ occurs not between consecutive indices but every $\alpha$ indices. In particular, we find solving the geometric and arithmetic-geometric sums of $a$ and $b$ that $\frac{b}{a} = \frac{M}{1-g(\alpha,M)} + u(\alpha)$
where $g(\alpha,M) = (2^{1/\alpha})^{-M-1}$ is a logarithmically decaying function of $M$ and $\alpha$. $u(\alpha) = 1/(1-2^{1/\alpha})$ is non-linear but slowly growing function of $\alpha$. Hence, an approximate GQ can operate as a bucketed-queue where indices start from  $I_0$ where $g(\alpha,M_0)\approx0$ and end at $I_{max}$ where $2^{f(I_{max})}$ can be precisely represented in the CPU word used to represent $a$ and $b$. In this case, there is a constant shift in the value  $ceil(b/a)$ that is calculated by $u(\alpha)$. 
For instance, consider an approximate queue with an $\alpha$ of 16. The function 
$g(\alpha,M)$ decays to near zero at $M=124$ making the shift $u(\alpha) = 22$. Hence, $I_0 = 124$ and $I_{max} =647$ which allows for the creation of an approximate queue that can handle 523 buckets.
%number of buckets $M=200$ and we use .
Note that this configuration results in an exact queue only when all buckets between $I_0$ and $I_{max}$ are nonempty. However, error is introduced when some elements are missing. In Section \ref{sec:eval_2}, we show the effect of this error through extensive experiments; more examples are shown in Appendix~\ref{app:example}.

%\vspace{-0.1in}

Typical scheduling policies (e.g., timestamp-based shaping, Least Slack Time First, and Earliest Deadline First) will  generate priority values for packets that are uniformly distributed over priority levels. For such scenarios, the approximate gradient queue will have zero error and extract the minimum element in one step. This is clearly not true for \textit{all} scheduling policies (e.g., strict priority will probably have more traffic for medium and low level priorities compared to high priority). For cases where the index suggested by the function is of an empty bucket, we perform linear search until we find a nonempty bucket. {\color{black} Moreover, for a cases of a moving range, a circular approximate queue can be implemented as with cFFS.}

Approximate queues have been used before for different use cases. For instance, Soft-heap \cite{chazelle2000soft} is an approximate priority queue with a bounded error that is inversely proportional to the overhead of insertion. In particular, after $n$ insertions in a soft-heap with an error bound $0~<~\epsilon~\leq~1/2$, the overhead of insertion is $O(\log (1/\epsilon))$. 
%Soft-heap can be used to calculate the exact median optimally, however, it is not very useful for packet scheduling as our focus in on
Hence, \verb|ExtractMin| operation which can have a large error under Soft-heap. Another example is the RIPQ which is  was developed for caching \cite{tang2015ripq}. RIPQ relies on a bucket-sort-like approach. However, the RIPQ implementation is suited for static caching, where elements are not moved once inserted, which makes it not very suitable for the dynamic nature of packet scheduling.

\vspace{-0.1in}

\subsection{Flexibility in Eiffel}
\label{sec:flex}
Our second objective is to deploy flexible schedulers that have full expressive power to implement a wide range of scheduling policies. {\color{black}Our goal is to provide the network operator with a compiler that takes as input policy description and produces an initial implementation of the scheduler using the building blocks provided in the previous section.}
Our starting point is the work in PIFO which develops a model for programmable packet scheduling \cite{pifo}.  PIFO, however, suffers from several drawbacks, namely: 1) it doesn't support reordering packets already enqueued based on changes in their flow ranking, 2) it does not support ranking of elements on packet dequeue, and 3) it does not support shaping the output of the scheduling policy. In this section, we show our augmentation of the PIFO model to enable a completely flexible programming model in Eiffel. We address the first two issues by adding programming abstractions to the PIFO model and we address the third problem by enabling arbitrary shaping with Eiffel by changing how shaping is handled within the PIFO model. {\color{black}We discuss the implementation of an initial version of the compiler in Section~4.}

\vspace{-0.1in}

%\subsubsection{Packet Scheduler Programming: A Primer}
\subsubsection{PIFO Model Extensions}

%Our work augments an existing scheduler programming model to create a more expressive model. 
Before we present our new abstractions, we review briefly the PIFO programming model \cite{pifo}. The model relies on the Push-In-First-Out (PIFO) conceptual queue as its main building block. In programming the scheduler, the PIFO blocks are arranged to implement different scheduling algorithms.

The PIFO programming model has three abstractions: 1) scheduling transactions, 2) scheduling trees, and 3) shaping transactions. A scheduling transaction represents a single ranking function with a single priority queue. Scheduling trees are formed by connecting scheduling transactions, where each node's priority queue contains an ordering of its children. The tree structure allows incoming packets to change the relative ordering of packets belonging to different policies. Finally, a shaping transaction can be attached to any non-root node in the tree to enforce a rate limit on it. There are several examples of the PIFO programming model in action presented in the original paper \cite{pifo}. The primitives presented in the original PIFO model capture scheduling policies that have one of the following features: 1) distinct packet rank enumerations, over a small range of values (e.g., strict priority), 2) per-packet ranking over a large range of priority values (e.g., Earliest Deadline First \cite{liu1973scheduling}), and 3) hierarchical policy-based scheduling (e.g., Hierarchical Packet Fair Queuing \cite{hierarchicalFQ}).

Eiffel augments the PIFO model by adding two additional scheduler primitives. The first primitive is {\em per-flow ranking and scheduling} where the rank of all packets of a flow depend on a ranking that is a function of  the ranks of all packets enqueued for that specific flow. We assume that a sequence of packets that belong to a single flow should not be reordered by the scheduler. Existing PIFO primitives keep per-flow state but use them to rank each packet individually where an incoming packet for a certain flow does not change the ranking of packets already enqueued that belong to the same flow. The per-flow ranking extension keeps track of that information along with a queue per flow for all packets belonging to that flow. A single PIFO block orders flows, rather than packets, based on their rank. The second primitive is {\em on-dequeue scheduling} where incoming and outgoing packets belonging to a certain flow can change the rank of all packets belonging to that flow on enqueue and dequeue. %and 3) support for arbitrary shaping (i.e., arbitrary assignment of per-flow or aggregate rate limits). %Furthermore, Eiffel replaces the PIFO conceptual queue with a priority queue. 

%\subsubsection{PIFO Model Extensions}
%\label{sec:prog_model}

%\textbf{Per-flow ranking:} 
%We assume that a sequence of packets that belong to a single flow should not be reordered by the scheduler. 
%Hence, our model supports a flexible definition of a \textit{flow}. A flow can be a group of TCP and UDP flows or a single transport layer flow. 

%However, scheduling algorithms can require this operation where the rank of all packets of a flow depend on a ranking that is a function of all packets enqueued for that specific flow. For instance, pFabric ranks packets based on their remaining size of their flows. However, this causes starvation for packets at the beginning of the flow. Hence, all packets belonging to the flow that has a packet with minimum flow size are advanced to the head of the queue (i.e., not just the packet with the minimum flow size). This scheduling algorithm cannot be modeled using a per-packet ranking only \cite{pifo}. However, it can be modeled using a per-flow ranking primitive where flows, as FIFO queues, are the elements kept in a priority queue and their priority level is calculated as a function of packets enqueued. 

%\textbf{On-dequeue ranking:} 
%All scheduling primitives presented by the PIFO model along with per-flow ranking are applied on enqueue. However, some complicated queues require ranking packets or flows on enqueue and dequeue.  For instance, 

\begin{figure}[!t]
{\small
%\begin{lstlisting}[language=Python, caption={.},captionpos=b,label={lst:lst1}]
\begin{python}
#On enqueue of packet p of flow f:
f.rank = f.len
#On dequeue of packet p of flow f:
f.rank = f.len
\end{python}
%\end{lstlisting}
}
\vspace{-0.1in}
\caption{Example of implementation of Longest Queue First (LQF)}
\label{fig:listing_lqf}
\vspace{-0.1in}
\end{figure}
The two primitives can be integrated in the PIFO model. All flows belonging to a \textit{per-flow} transaction are treated as a single flow by scheduling transactions higher in the hierarchical policy. Also note that every individual flow in the flow-rank policy can be composed of multiple flows that are scheduled according to per packet scheduling transactions. Figure~\ref{fig:listing_lqf} is an example of how both primitives are expressed in a PIFO programming model.  The example implements Longest Queue First (LQF) which requires flows to be ranked based on their length which changes on enqueue and dequeue.  The example allows for setting a rank for flows, not just packets (i.e., \verb|f.rank|). It also allows for programming actions on both enqueue and dequeue. {\color{black}We realize that this specification requires tedious work to describe a complex policy that handles thousands of different flows or priorities. However, this specification provides a direct mapping to the underlying priority queues. We believe that defining higher level programming languages describing packet schedulers as well as formal description of the expressiveness of the language to be topics for future research.}

\vspace{-0.1in}
\subsubsection{Arbitrary Shaping}
\label{sec:shaping}
\vspace{-0.1in}

A flexible packet scheduler should support any scheme of bandwidth division between incoming flows. Earlier work on flexible schedulers either didn't support shaping at all (e.g., OpenQueue) or supported it with severe limitations (e.g., PIFO). We allow for arbitrary shaping by decoupling work conserving scheduling from shaping. A natural approach to this decoupling is to allow any flow or group of flows to have a shaper associated with them. This can be achieved by assigning a separate queue to the shaped aggregate whose output is then enqueued into its proper location in the scheduling hierarchy. However, this approach is extremely inefficient as it requires a queue per rate limit, which can lead to increased CPU and memory overhead. We improve the efficiency of this approach by leveraging recent results that show that any rate limit can be translated to a timestamp per packet, which yields even better adherence to the set rate than token buckets \cite{carousel}. Hence, we use only one shaper for the whole hierarchy which is implemented using a single priority queue.

\begin{figure}[!t]
    \centering
%    \begin{minipage}{0.48\textwidth}
\includegraphics[width=0.9\linewidth]{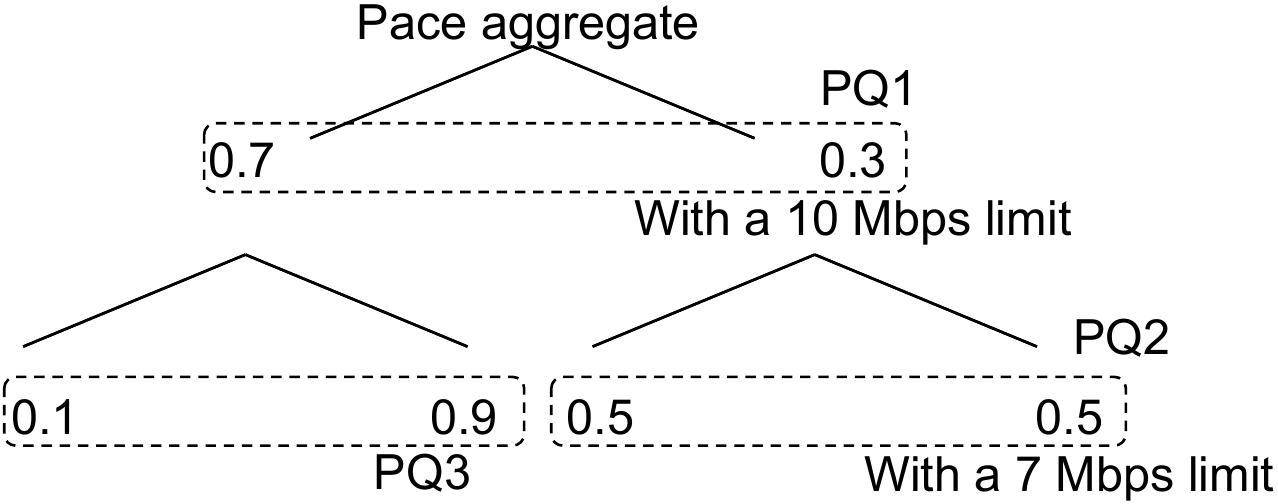}
\vspace{-0.1in}
\caption{Example of a policy that imposes two limits on packets the belong to the rightmost leaf.}
\label{fig:shaping_policy}
\vspace{-0.1in}
%\end{minipage}\hfill
%    \begin{minipage}{0.48\textwidth}
\end{figure}

\begin{figure}[!t]
\centering
\includegraphics[width=0.9\linewidth]{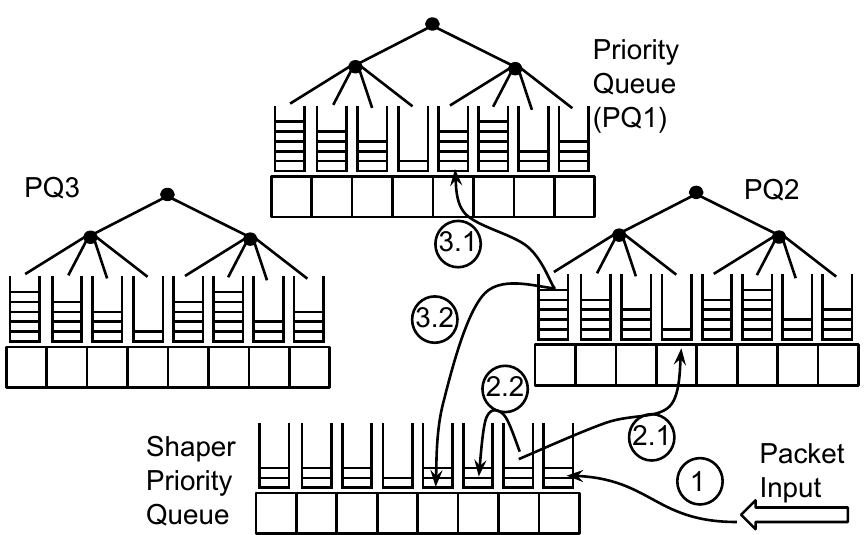}
\vspace{-0.1in}
\caption{A diagram of the implementation of the example in Figure~\ref{fig:shaping_policy}.}
\label{fig:shaping_policy_2}
%    \end{minipage} 
\vspace{-0.1in}
\end{figure}

As an example, consider the hierarchical policy in Figure~\ref{fig:shaping_policy}. {\color{black}Each node represents a policy-defined flow with the root representing the aggregate traffic. Each node has a share of its parent's bandwidth, defined by the fraction in the figure. Each node can also have a policy-defined rate limit. In this example, we have a rate limit at a non-leaf node and a leaf node. Furthermore, we require the aggregate traffic to be paced. We map the hierarchical policy in Figure~\ref{fig:shaping_policy} to its priority-queue-based realization in Figure~\ref{fig:shaping_policy_2}.} Per the PIFO model, each non-leaf node is represented by a priority queue. {\color{black} Per our proposal, a single shaper is added to rate limit all packets according to all policy-defined rate limits.} 

{\color{black}To illustrate how this single shaper works,}
%In the general case, a packet can be dequeued from a priority queue in a hierarchy, enqueued in the shaper, to be enqueued in another priority queue in the hierarchy. As shown in Figure~\ref{fig:shaping_policy}, 
consider packets belonging to the rightmost leaf policy. {\color{black} We explore the journey of packets belonging to that leaf policy through the different queues, as shown in Figure~\ref{fig:shaping_policy_2}. These packets will be enqueued to the shaper with timestamps set based on a 7 Mbps rate to enforce the rate on their node (step 1).  Once dequeued from the shaper, each packet will be enqueued to PQ2 (step 2.1) and the shaper according to the 10 Mbps rate limit (step 2.1). After the transmission time of a packet belonging to PQ2 is reached, which is defined by the shaper, the packet is inserted in both the root's (PQ1) priority queue (3.1) and the shaper according to the pacing rate (3.2). When the transmission time, calculated based on the pacing rate, is reached the packet is transmitted.
%when dequeued from PQ2, each packet will be enqueued in the shaper again with timestamps set based on a 10 Mbps limit which will be shared with flows from the second to right policy flows, before being enqueued to PQ1. A sketch of this process is shown in Figure~\ref{fig:shaping_policy_2}, where a single shaper (at the bottom) is used to shape packets for all nodes in the policy tree. 
To achieve this functionality, each packet holds a pointer to the priority queue they should be enqueued to. This pointer avoids searching for the queue a packet should be enqueued to. Note that having the separate shaper allows for specifying rate limits on any node in the hierarchical policy (e.g., the root and leaves) which was not possible in the PIFO model, where shaping transactions are tightly coupled with scheduling transactions.}  %This is facilitated by bucketed queues which support move and remove operations in $O(1)$ given that an element in the queue keeps track of its index in the queue. This is unlike typical priority queues which require operations like \verb|make_heap| if a random element is removed.

%\subsection{Eiffel System Optimization}
\vspace{-0.1in}

\section{Eiffel Implementation}
\label{sec:impl}
\vspace{-0.1in}
Packet scheduling is implemented in two places in the network: 1) hardware or software switches, and 2) end-host kernel. We focus on the software placements (kernel and userspace switches) and show that Eiffel can outperform the state of the art in both settings. We find that userspace and kernel implementations of packet scheduling face significantly different challenges as the kernel operates in an event-based setting while userspace operates in a busy polling setting. We explain here the differences between both implementations and our approach to each. We start with our approach to policy creation.

\textbf{Policy Creation:} We extend the existing PIFO open source model to configure the scheduling algorithm \cite{pifo, pifo_github}. The existing implementation represents the policy as a graph using the DOT description language and translates the graph into C++ code. We rely on the cFFS for our implementation, unless otherwise stated. This provides an initial implementation which we tune according to whether the code is going to be used in kernel or userspace. We believe automating this process can be further refined, but the goal of this work is to evaluate the performance of Eiffel algorithms and data structures.

\textbf{Kernel Implementation}  We implement Eiffel as a qdisc \cite{hubert2002linux} kernel module that implements enqueue and dequeue functions and keeps track of the number of enqueued packets. The module can also set a timer to trigger dequeue. Access to qdiscs is serialized through a global qdisc lock. In our design, we focus on two sources of overhead in a qdisc: 1) the overhead of the queuing data structure, and 2) the overhead of properly setting the timer. Eiffel reduces the first overhead by utilizing one of the proposed data structures to reduce the cost of both enqueue and dequeue operations. The second overhead can be mitigated by improving the efficiency of finding the smallest deadline of an enqueued packet. This operation of \verb|SoonestDeadline()| is required to efficiently set the timer to wake up at the deadline of the next packet. Either of our supported data structures can support this operation efficiently as well.

\textbf{Userspace Implementation} We implement Eiffel in the Berkeley Extensible Software Switch (BESS, formerly SoftNIC \cite{softnic}). BESS represents network processing elements as a pipeline of modules. BESS is busy polling-based where a set of connected modules form a unit of execution called a task. A scheduler tracks all tasks and runs them according to assigned policies. Tasks are scheduled based on the amount of resources (CPU cycles or bits) they consume. 
Our implementation of Eiffel in BESS is done in self-contained modules. 

We find that two main parameters determine the efficiency of Eiffel in BESS: 1) batch size and 2) queue size. Batching is already well supported in BESS as each module receives packets in batches and passes packets to its subsequent module in a batch. However, we find that batching per flow has an intricate impact on the performance of Eiffel. For instance, with small packet sizes, if no batching is performed per flow, then every incoming batch of packets will activate a large number of queues without any of the packets being actually queued (due to small packet size) which increases the overhead per packet (i.e., queue lookup of multiple queues rather than one). This is not the case for large packet sizes where the lookup cost is amortized over the larger size of the packet improving performance compared to batching of large packets. Batching large packets results in large queues for flows (i.e., large number of flows with large number of enqueued packets). We find that batching should be applied based on expected traffic pattern. For that purpose, we setup \verb|Buffer| modules per traffic class before Eiffel's module in the pipeline when needed. We also perform output batching per flow in units of 10KB worth of payload which was suggested as a good threshold that does not affect fairness at a macroscale between flows \cite{billaud2013hclock}. We also find that limiting the number of packets enqueued in Eiffel can significantly affect the performance of Eiffel in BESS. This is something we did not have to deal with in the kernel implementation because of TCP Small Queue (TSQ) \cite{tsq} which limits number of packets per TCP flow in kernel stack. We limit the number of packets per flow to 32 packets which we find, empirically, to maintain performance. 
\vspace{-0.1in}

\section{Evaluation}
\label{sec:eval}
\vspace{-0.1in}
%XXX Do we need roadmap paragraph here?

We evaluate Eiffel through three use cases. Then, we perform microbenchmarks to evaluate different data structures and measure the effect of approximation on network wide objectives. Finally, we present a guide for picking a data structure.

\vspace{-0.1in}

\subsection{\algoname{} Use Cases}
\vspace{-0.1in}

\textbf{Methodology:} We evaluate our kernel and userspace implementation through a set of use cases each with its corresponding baseline.  We implement three common use cases, one  in kernel and two in userspace.  In each use case, we evaluated Eiffel's scheduling behavior as well as its CPU performance as compared to the baseline. {\color{black}The comparison of scheduling behavior was done by comparing aggregate rates achieved as well as order of released packets.} However, we only report CPU efficiency results as we find that Eiffel matches the scheduling behavior of the baselines. %Hence, we do not list all possible evaluation scenarios for each system and focus only on CPU efficiency.

{\color{black}A key aspect of our evaluation is determining the metrics of comparisons in kernel and userspace settings. The main difference is that a kernel module can support line rate by using more CPU. This requires us to fix the packet rate we are evaluating at and look at the CPU utilization of different scheduler implementations. On the other hand, a userspace implementation relies on busy polling on one or more CPU cores to support different packet rates. Hence, in the case of userspace, we fix the number of cores used, to one core unless otherwise is stated, and compare the different scheduler implementations based on the maximum achievable rate. The following use cases cover the two settings, presenting the proper evaluation metrics for each setting. We find that for all use cases Eiffel reduces the CPU overhead per packet which leads to lower CPU utilization in kernel settings and higher maximum achievable rate in userspace settings.}

\vspace{-0.1in}

\subsubsection{Use Case 1: Shaping in Kernel}
\vspace{-0.08in}

Traffic shaping (i.e., rate limiting and pacing) is an essential operation for efficient utilization \cite{need_pacing} and correct operation of modern protocols (e.g., both TIMELY \cite{timely} and BBR \cite{bbr} require per flow pacing). Recently, it has been shown that the canonical kernel shapers (i.e., FQ/pacing \cite{fq} and HTB qdiscs \cite{linuxhtb}) are inefficient due to reliance on inefficient data structures; they are outperformed by the userspace-based implementation in Carousel \cite{carousel}. To offer a fair comparison we implement all systems in the kernel.

We implement a rate limiting qdisc whose functionality matches the rate limiting features of the existing FQ/pacing qdisc \cite{fq}. In particular, we allow a flow to set a \verb|SO_MAX_PACING_RATE| parameter which is a rate limit for the flow.
%\textbf{[xxx should we redo experiments to support Start-Time Fair Queueing in order to match FQ completely?]}. 
Furthermore, we also calculate a pacing rate limit for each TCP flow in case of unspecified maximum rate. We also implement a Carousel-based qdisc, in accordance with recommendations of the authors of the Carousel paper. We implement a qdisc  where all packets are queued in a timing wheel. A timer fires every time instant (according to the granularity of the timing wheel) and checks whether it has packets that should be sent.  

\begin{figure}[t]
\centering
\includegraphics[width=0.95\linewidth]{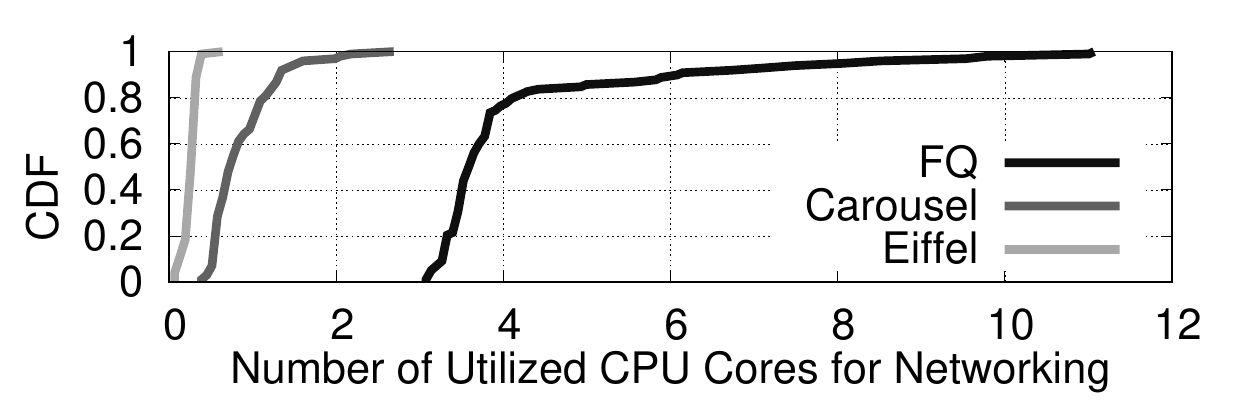}
\vspace{-0.1in}
\caption{A comparison between the CPU overhead of the networking stack using FQ/pacing, Carousel, and Eiffel.
%showing that Eiffel improves performance by up to 4x compared to Carousel and 20x compared to FQ/pacing.
}
\label{fig:kernel_network}
\vspace{-0.1in}
\end{figure}

\begin{figure}[t]
\centering
\includegraphics[width=0.492\linewidth]{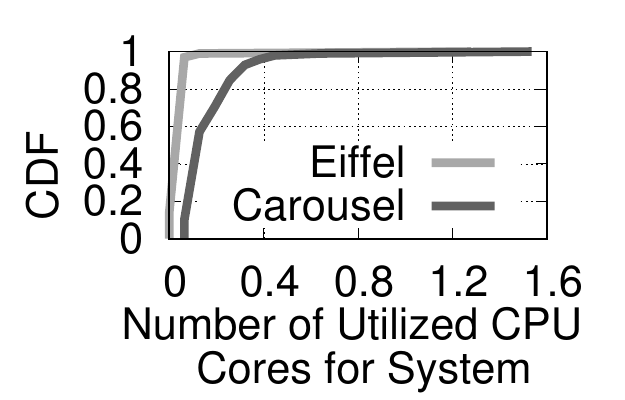}
\includegraphics[width=0.492\linewidth]{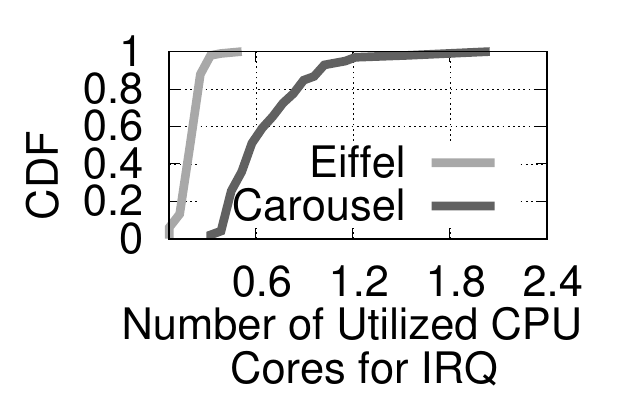}
\vspace{-0.1in}
\caption{A Comparison between detailed CPU utilization of Carousel and Eiffel in terms of system processes (left) and soft interrupt servicing (right).}
\label{fig:kerne_network_detailed}
\vspace{-0.1in}
\end{figure}

We implemented \algoname{} as a qdisc. 
%Furthermore, we move calculations of packet transmission time to TCP rather than performing all accounting inside the Qdisc. 
%Our implementation relies on a circular FFS-based. 
The queue is configured with 20k buckets with a maximum horizon of 2 seconds and  only the shaper is used. We implemented the qdisc in kernel v4.10. We modified only \verb|sock.h| to keep the state of each socket allowing us to avoid having to keep track of each flow in the qdisc. We conduct experiments for egress traffic shaping between two servers within the same cluster in Amazon EC2. We use two \verb|m4.16xlarge| instances equipped with 64 cores and capable of sustaining 25 Gbps. We use \verb|neper| \cite{neper_github} to generate traffic with a large number of TCP flows. In particular, we generate traffic from 20k flows and use \verb|SO_MAX_PACING_RATE| to rate limit individual flows to achieve a maximum aggregate rate of 24 Gbps. This configuration constitutes a worst case in terms of load for all evaluated qdiscs as it requires the maximum amount of calculations. We measure overhead in terms of the number of cores used for network processing which we calculate based on the observed fraction of CPU utilization. Without \verb|neper| operating, CPU utilization is zero, hence, we attribute any CPU utilization during our experiments to the networking stack, except for the CPU portion attributed to userspace processes. We track CPU utilization using \verb|dstat|. We run our experiments for 100 seconds and record the CPU utilization every second. This continuous behavior emulates the behavior handled by content servers which were used to evaluate Carousel \cite{carousel}. %We measure RTT during experiments using ping and record its RTT every second.  

Figure~\ref{fig:kernel_network} shows the overhead of all three systems. It is clear that Eiffel is superior, outperforming FQ by a median 14x and Carousel by 3x. We find the overhead of FQ to be consistent with earlier results \cite{carousel}. This is due to its complicated data structure which keeps track internally of active and inactive flows and requires continuous garbage collection to remove  old inactive flows. Furthermore, it relies on RB-trees which increases the overhead of reordering flows on every enqueue and dequeue. To better understand the comparison with Carousel, we look at the breakdown of the main components of CPU overhead, namely overhead spent on \textit{system processes} and \textit{servicing software interrupts}. Figure~\ref{fig:kerne_network_detailed} details the comparison. We find that the main difference is in the overhead introduced by Carousel in firing timers at constant intervals while Eiffel can trigger timers exactly when needed (Figure~\ref{fig:kerne_network_detailed} right). The overhead of the data structures in both cases introduces minimal overhead in system processes (Figure~\ref{fig:kerne_network_detailed} left). %We observed no difference in scheduling behavior of all qdiscs in terms of latency or adherence to target rate limits.
\vspace{-0.1in}
\subsubsection{Use Case 2: Flow Ranking Policy in Userspace}
\vspace{-0.1in}

\begin{figure}[!t]
{
%\begin{lstlisting}[language=Python, caption={Example of implementation of pFabric with starvation prevention.},captionpos=b,label={lst:lst1}]
\begin{python}
#On enqueue of packet p of flow f:
f = flow(p)
# Shaping transaction ranking
f.r_rank = f.r_rank + \
    p.size / f.reservation
f.l_rank = f.l_rank + p.size / f.limit
# Scheduling transaction ranking
f.s_rank = f.s_rank + p.size / f.share

\end{python}
%\end{lstlisting}
}
\vspace{-0.1in}
\caption{Implementation of hClock in Eiffel.}
\vspace{-0.1in}
\label{fig:listing_hclock}
\end{figure}

%One of the advantages of Eiffel is maintaining efficiency while providing the flexibility to represent a wide range of scheduling policies. 
We demonstrate the efficiency of Eiffel's Per-Flow Rank programming abstraction by implementing hClock \cite{billaud2013hclock}, a hierarchical packet scheduler that supports limits, reservations, and proportional shares. hClock is the main implementation of scheduling in NetIOC in VMware's vSphere \cite{vmware_2} and relies on min heaps the incur $O(\log n)$  overhead per packet batch \cite{billaud2013hclock}. We implement hClock based on its original specs and compare that with hClock using Eiffel. Both versions are implemented in BESS; each in a self-contained single module. We also attempt to replicate hClock's behavior using the traffic control ({\em tc}) mechanisms in BESS. However, this requires instantiating a module corresponding to every flow which incurs a large overhead for a large number of flows.

Figure~\ref{fig:listing_hclock} shows the implementation of hClock using PIFO model along with Eiffel extensions. hClock relies on three ranks to determine order of packets: \verb|r_rank| corresponds to the reservation of the flow which is the ranking used to order flows in the arbitrary shaper, \verb|s_rank| is calculated based on the share of a flow compared to its siblings in the tree and this rank is used in the scheduling transaction, {\color{black}and \verb|l_rank| is the rank of a packet based on the rate limit of its flow. The packet dequeued from the shaper belongs to the flow with the smallest \verb|r_rank| smaller than the current time whose \verb|l_rank| is also smaller than current time. Allowing each flow to have multiple ranks as such creates a complication not capture by the Eiffel syntax. Hence, this complication requires adding some code manually. However, it does not hurt efficiency.}

\begin{figure}[t]
\centering
\includegraphics[width=0.98\linewidth]{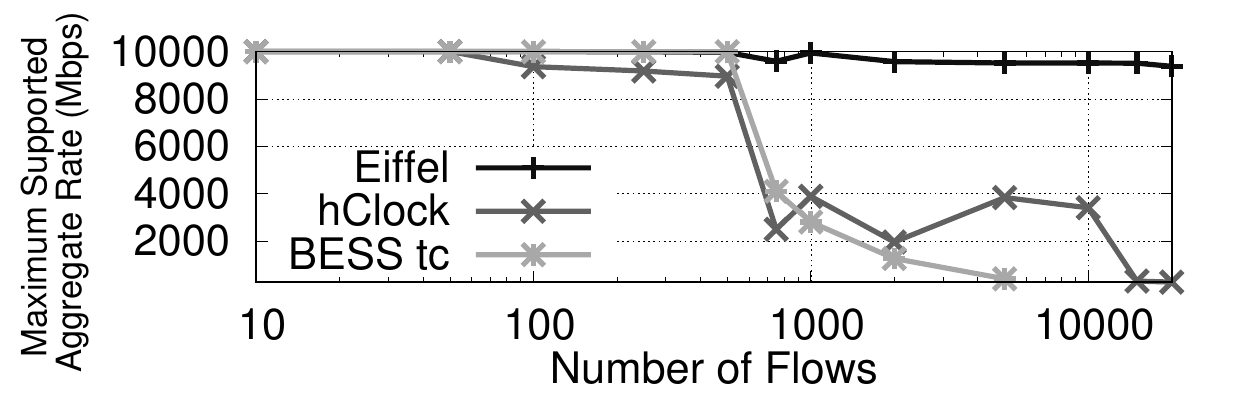}
\includegraphics[width=0.98\linewidth]{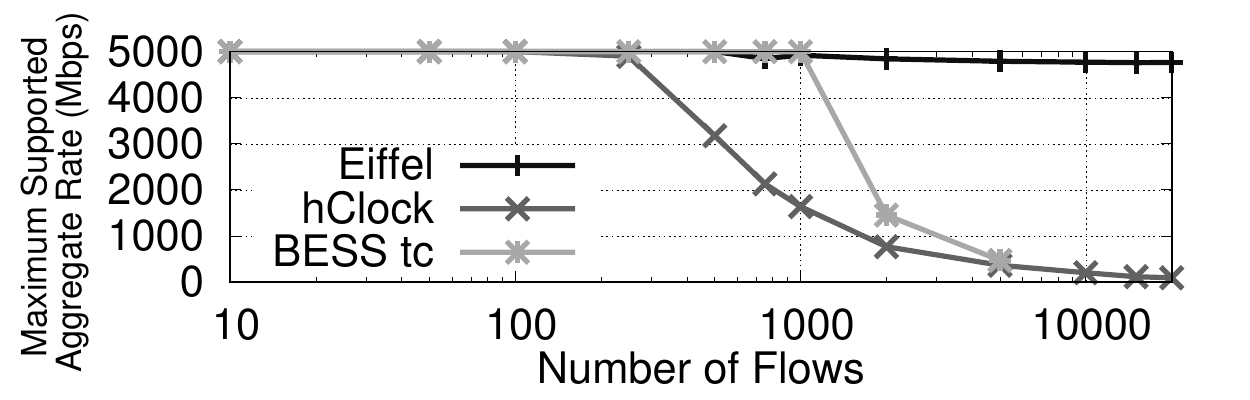}
\vspace{-0.1in}
\caption{Comparison between maximum supported aggregate rate limit (top) and behavior at a rate limit of 5 Gbps (bottom) for hClock, Eiffel's implementation of hClock, and BESS tc on a single core with no batching.}
\label{fig:bess_batching}
\vspace{-0.1in}
\end{figure}

\begin{figure}[t]
\centering
\includegraphics[width=0.98\linewidth]{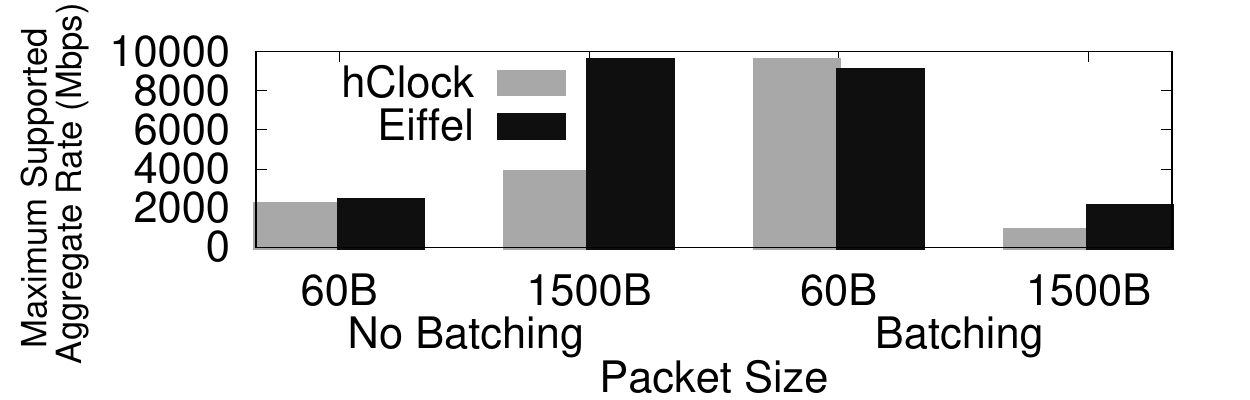}
\vspace{-0.1in}
\caption{Effect of batching and packet size on throughput for both Eiffel and hClock for 5k flows.}
\vspace{-0.1in}
\label{fig:bess_flows_hclock}
\end{figure}

We conduct our experiments on lab machines equipped with Intel Xeon CPU ES-1620 with 8 cores, 64GB of memory, and Intel X520-SR2 dual port NICs capable of an aggregate of 20 Gbps. We use a simple packet generator implemented in BESS and a simple round robin annotator to distribute packets over traffic classes. We change the number of traffic classes (i.e., flows) and measure the maximum aggregate rate. To demonstrate the efficiency of Eiffel, we conduct all experiments on a single core. We run each experiment for 10 seconds and plot the average observed throughput at the port of the sender which matches the rate at the receiver.

Figure~\ref{fig:bess_batching} shows the results of our experiments for varying number of flows. Packets are 1500B which is MTU size. We run experiments at a line rate of 10 Gbps and with a rate limit of 5 Gbps. It is clear that Eiffel can support line rate at up to 40x the number of flows compared to hClock at both aggregate rate limits and even larger advantage compared to tc in BESS. We find that all three implementations scale well with large number of flows when operating on two cores. However, it is clear that using Eiffel results in significant savings.

%We also find that the implementation of Eiffel scales well with the number of cores allowing it to double the number of flows it can schedule while retaining line rate. It can also rate limit as accurately and at even larger number of flows.

We also investigate the combined effects of packet size and per-flow batching (Figure~\ref{fig:bess_flows_hclock}). As discussed earlier, performance for small packet sizes requires batching to improve performance due to memory efficiency (i.e., all packets within the same batch are queued in the same place). However, for large packet sizes the value of batching is offset by the overhead of large queue lengths. Our results validate this hypothesis showing that when per-flow batching is applied with small packet sizes both hClock and Eiffel can maintain performance close to line rate with Eiffel performing worse than hClock by 5-10\%. For per-flow batching with large packet sizes, per-flow queues are larger causing performance degradation. We note that the typical behavior of a network will avoid per-flow batching, as it might cause unnecessary latency. In that case, Eiffel outperforms hClock and can sustain line rate for large packet sizes, which should also be the common case.

%As we show later, for short queues, it is likely that binary heaps will perform better as they require less memory and have a comparable number of operations to a integer priority queue.

%\subsubsection{Least Slack Time First in Userspace}

%\textbf{[xxx add one more scheduling algorithm that is not typical in networks]
%}

\begin{figure}[!t]
{
\begin{python}
#On enqueue of packet p of flow f:
f.rank = min(p.rank, f.rank)
#On dequeue of packet p of flow f:
f.rank = min(p.rank, f.front().rank)
\end{python}
}
\vspace{-0.1in}
\caption{Implementation of pFabric in Eiffel.}
\vspace{-0.1in}
\label{fig:listing_pfabric}
\end{figure}
%{\small
%\begin{lstlisting}[language=Python, caption={Example of implementation of pFabric with starvation prevention.},captionpos=b,label={lst:lst1}]
%#On enqueue of packet p of flow f:
%f.rank = min(p.rpt, f.rank)
%#On dequeue of packet p of flow f:
%f.rank = f.extract_min_p_rpt()
%\end{lstlisting}
%}
%\vspace{-0.2in}
%\end{figure}

%Listing~\ref{lst:lst1} is an example of how such ranking can be implemented in a PIFO programming model hierarchy. The example highlights the two abstractions. In particular, it allows for setting a rank for flows, not just packets (i.e., \verb|f.rank|). It also allows for programming actions on both enqueue and dequeue.

\vspace{-0.1in}

\subsubsection{Use Case 3: Least/Largest X First in Userspace}
\vspace{-0.1in}

\begin{figure}[t]
\centering
\includegraphics[width=0.98\linewidth]{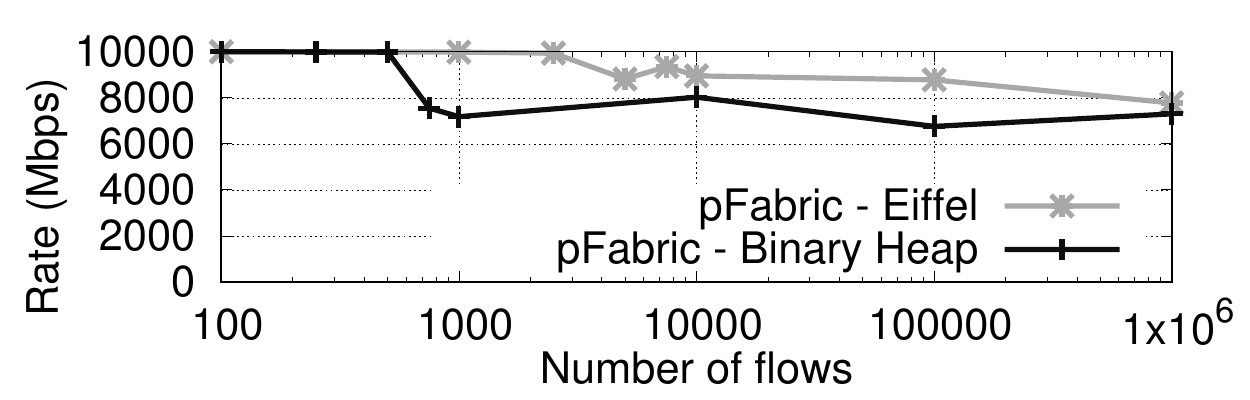}
\vspace{-0.1in}
\caption{Performance of pFabric implementation using cFFS and a binary heap showing  Eiffel sustaining line rate at 5x number of flows.}
\vspace{-0.1in}
\label{fig:bess_pfabric}
\end{figure}

One of the most widely used patterns for packet scheduling is ordering packets such that the flow or packet with the least or most of some feature exits the queue first. Many examples of such policies have been promoted including Least Slack Time First (LSTF) \cite{UPS}, Largest Queue First (LQF), and Shortest/Least Remaining Time First (SRTF). We refer to this class of algorithms as L(X)F. This class of algorithms is interesting as some of them were shown to provide theoretically proven desirable behavior. For instance, LSTF was shown be a universal packet scheduler that can emulate the behavior of any scheduling algorithm \cite{UPS}. Furthermore, SRTF was shown to schedule flows close to optimally within the pFabric architecture \cite{pfabric}. We show that Eiffel can improve the performance of this class of scheduling algorithms.

We implement pFabric as an instance of such class of algorithms where flows are ranked based on their remaining number of packets. Every incoming and outgoing packet changes the rank of all other packets belonging to the same flow, requiring on dequeue ranking. Figure~\ref{fig:listing_pfabric} shows the representation of pFabric using the PIFO model with per-flow ranking and on dequeue ranking provided by Eiffel. We also implemented pFabric using $O(\log n)$ priority queue based on a Binary Heap to provide a baseline. Both designs were implemented as queue modules in BESS. We used packets of size 1500B. Operations are all done on a single core with a simple flow generator. All results are the average of ten experiments each lasting for 20 seconds. Figure~\ref{fig:bess_pfabric} shows the impact of increasing the number of flows on the performance of both designs. It is clear that Eiffel has better performance. The overhead of pFabric stems from the need to continuously move flows between buckets which has $O(1)$ using bucketed queues while it has an overhead of $O(n)$ as it requires re-heapifying the heap every time. {\color{black} The figure also shows that as the number of flows increases the value of Eiffel starts to decrease as Eiffel reaches its capacity.}

\vspace{-0.1in}

\subsection{\algoname{} Microbenchmark}
\label{sec:eval_2}
\vspace{-0.08in}

Our goal in this section is evaluate the impact of different parameters on the performance of different data structures. We also evaluate the effect of approximation in switches on network-wide objectives. Finally, we provide guidance on how one should choose among the different queuing data structures within Eiffel, given specific scheduler user-case characteristics. To inform this decision we run a number of microbenchmark experiments. We start by evaluating the performance of the proposed data structures compared to a basic bucketed priority queue implementation. Then, we explore the impact of approximation using the gradient queue both on a single queue and at a large network scale through ns2-simulation. Finally, we present our guide for choosing a priority queue implementation. 

\textbf{Experiment setup:} We perform benchmarks using Google's benchmark tool \cite{benchmark_github}. We develop a baseline for bucketed priority queues by keeping track of non-empty buckets in a binary heap, we refer to this as BH. We ignore comparison-based priority queues (e.g., Binary Heaps and RB-trees) as we find that bucketed priority queues performs 6x better in most cases. We compare cFFS, approximate gradient queue (Approx), and BH. In all our experiments, the queue is initially filled with elements according to queue occupancy rate or average number of packet per bucket parameters. Then, packets are dequeued from the queue. Reported results  (i.e., y-axis of figures~\ref{fig:packet_per_bucket} and \ref{fig:queue_occu}) are in terms of million packets per seconds.

\begin{figure}[t]
\centering
\includegraphics[width=0.475\linewidth]{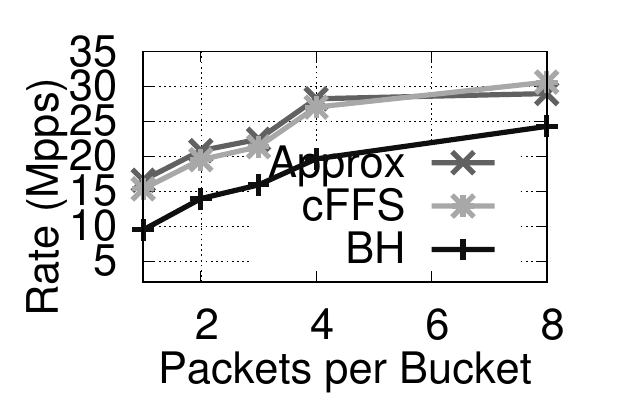}
\includegraphics[width=0.475\linewidth]{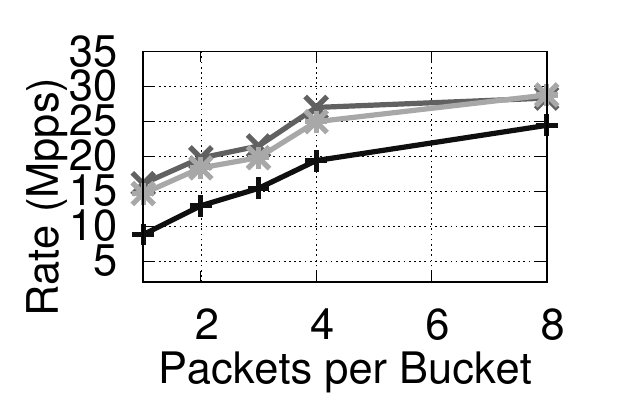}
\vspace{-0.1in}
\caption{Effect of number of packets per bucket on queue performance for 5k (left) and 10k (right) buckets.}
\vspace{-0.1in}
\label{fig:packet_per_bucket}
\end{figure}

\begin{figure}[t]
\centering
\includegraphics[width=0.475\linewidth]{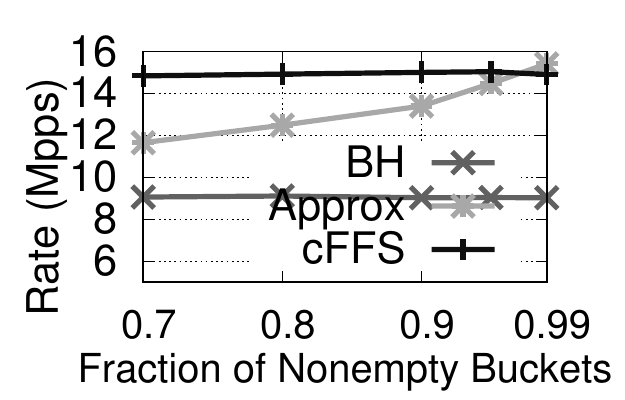}
\includegraphics[width=0.475\linewidth]{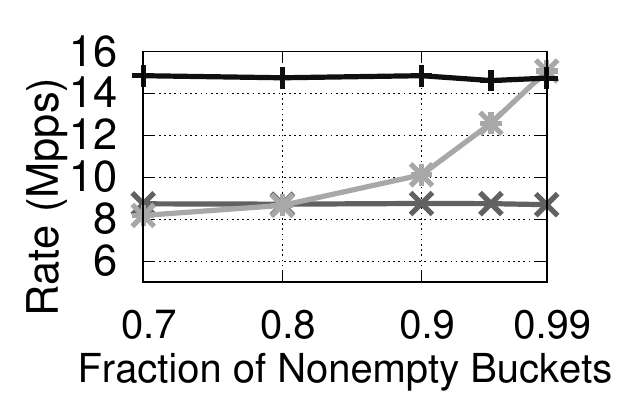}
\vspace{-0.1in}
\caption{Effect of queue occupancy on performance of Approximate Queue for 5k (left) and 10k (right) buckets.}
\vspace{-0.1in}
\label{fig:queue_occu}
\end{figure}

\begin{figure}[t]
\centering
\includegraphics[width=0.95\linewidth]{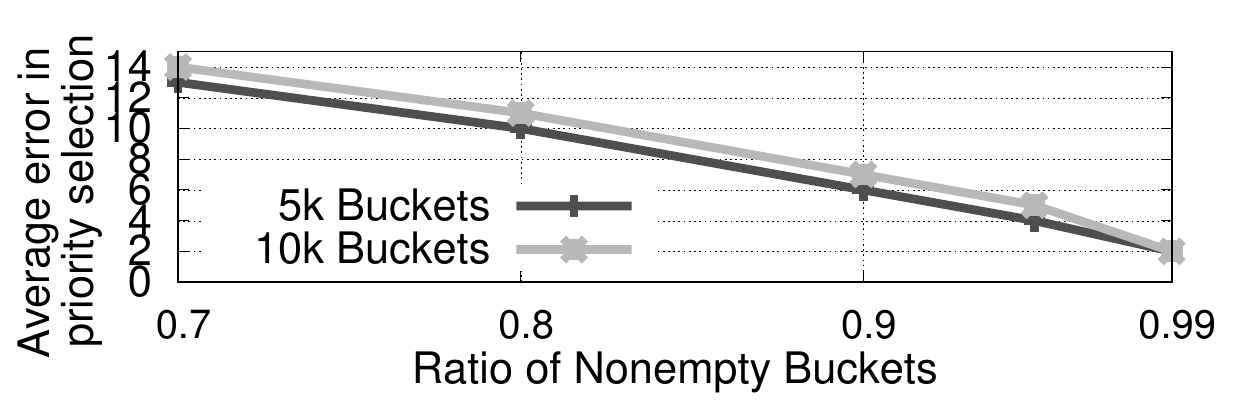}
\vspace{-0.1in}
\caption{Effect of having empty buckets on the error of fetching the minimum element for the approximate queue.}
\vspace{-0.1in}
\label{fig:error}
\end{figure}

\textbf{Effect of number of packet per bucket:}  The number of buckets configured in a queue is the main determining factor for the overhead of a bucketed queue. Note that this parameter controls queue granularity which is the priority interval covered by a bucket. High granularity (i.e., large number of buckets) implies a smaller number of packets per bucket for the same workload. Hence, the number of packets per bucket is a good proxy to the configured number of buckets. {\color{black} For instance, if we choose a large number of buckets with high granularity, the chance of empty buckets increases. On the other hand, if we choose a small number of buckets with coarser granularity, we get higher number of elements per bucket. This proxy is important because in the case of the approximate queue, the main factor affecting its performance is the number of empty buckets.}

Figure~\ref{fig:packet_per_bucket} shows the effect of increasing the average number of packets per bucket for all three queues for 5k and 10k buckets. For a small number of packets per bucket, which also reflects choosing a fine grain granularity, the approximate queue introduces up to 9\% improvement in performance in the case of 10k buckets. In such cases, the approximate queue function has zero error which makes it significantly better. As the number of the packets per bucket increases, the overhead of finding the smallest indexed bucket is amortized over the total number of elements in the bucket which makes FFS-based and approximate queues similar in performance.

\begin{figure}[t]
\centering
\includegraphics[width=0.95\linewidth]{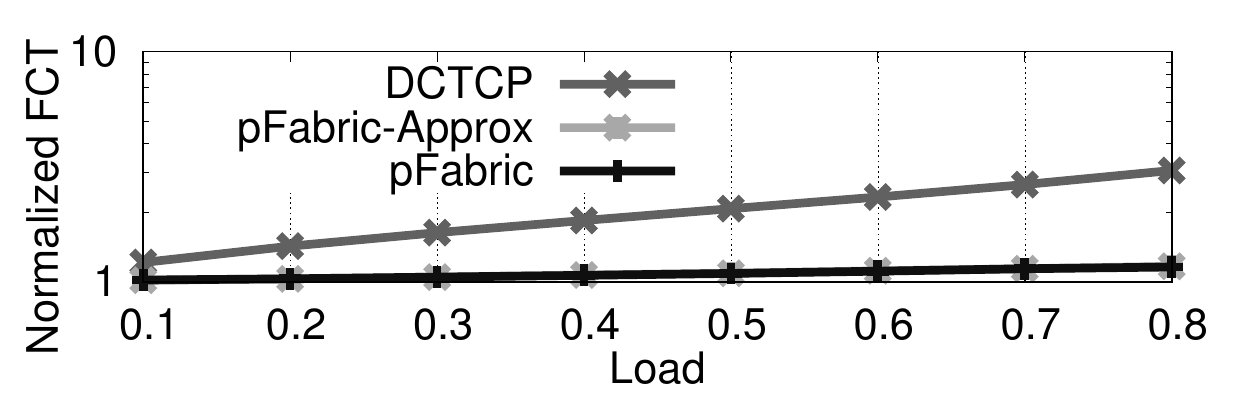}
\includegraphics[width=0.95\linewidth]{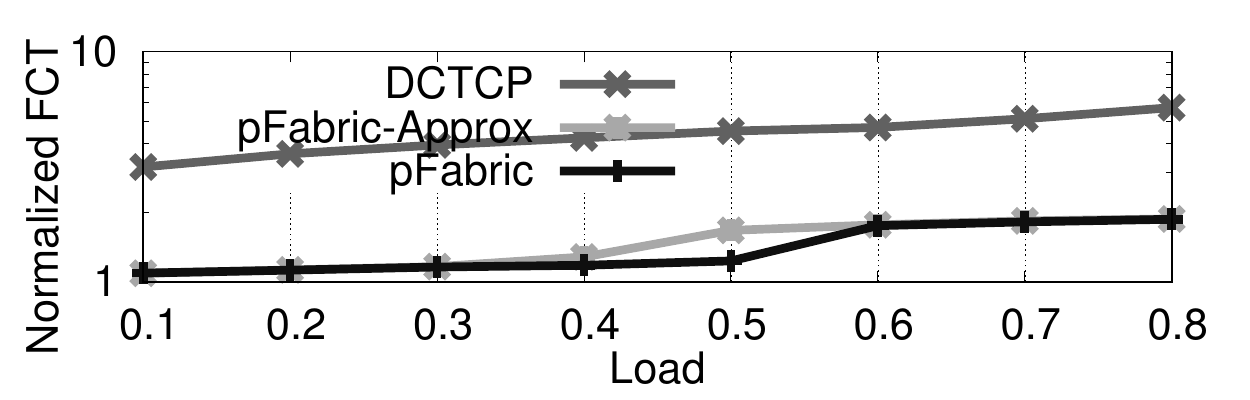}
\includegraphics[width=0.95\linewidth]{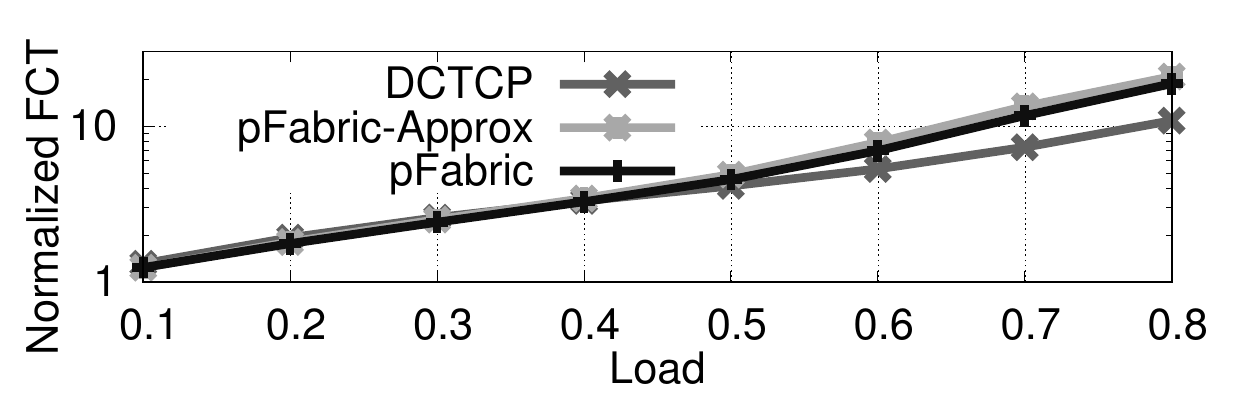}
\vspace{-0.1in}
\caption{Effect of using an Approximate Queue 
%in place of an exact priority queue 
on the performance of pFabric in terms of normalized flow completion times 
%in a 144-host simulation 
under different load characteristics: Average FCT for (0, 100kB] flow sizes, 99th percentile FCT for (0, 100kB] for sizes, and Average FCT for (10MB, $\inf$) flow sizes.}
\vspace{-0.1in}
\label{fig:pfabric}
\end{figure}

We also explore the effect of having empty buckets on the performance of the approximate queue. Empty buckets cause errors in the curvature function of the approximate queue which in turn trigger linear search for non-empty buckets. Figure~\ref{fig:packet_per_bucket}  shows throughput of the queue for different ratios of non-empty buckets. As expected, as the ratio increases the overhead decreases which improves the throughput of the approximate queue. Figure~\ref{fig:error} shows the error in the approximate queue's fetching of elements. As the number of empty buckets increases the error in the approximate queue is larger and the overhead of linear search grows. We suggest that cases where the queue is more than 30\% empty should trigger changes in the queue's granularity based on the queue's CPU performance and to avoid allocating memory to buckets that are not used.

The granularity of the queue determines the representation capacity of the queue. It is clear for our results that picking low granularity (i.e., high number of packets per bucket) yields better performance in terms of packets per second. On the other hand, from a networking perspective, high granularity yields exact ordering of packets. For instance, a queue with a granularity of 100 microseconds cannot insert gaps between packets that are smaller than 100 microseconds. Hence, we recommend configuring the queue's granularity such that each bucket has at least one packet. This can be determined by observing the long term behavior of the queue. We also note that this problem can be solved by having non-uniform bucket granularity which is dynamically set to achieve the result of at least one packet per bucket. We leave this problem for future work.

\textbf{Impact of Approximation on Network-wide Objectives:} A natural question is: how does approximate prioritization, at \textit{every} switch in a network, affect network-wide objectives? To answer that question, we perform simulations of pFabric, which requires prioritization at every switch. Our simulation are based on \verb|ns2| simulations provided by the authors of pFabric \cite{pfabric} and the plotting tools provided by the authors of QJump \cite{qjump}.  We change only the priority queuing implementation from a linear search-based priority queue to our Approximate priority queue and increase queue size to handle 1000k elements. We use DCTCP \cite{dctcp} as a baseline to put the result in context. Figure~\ref{fig:pfabric} shows a snapshot of results of the simulations of a  144 node leaf-spine topology. Due to space limitations, We show results for only web-search workload simulations which are based on clusters in Microsoft datacenters \cite{dctcp}. The load is varied between 10\% to 80\% of the load observed. We note that the setting of the simulations is not relevant for the scope of this paper, however, what is relevant is comparing the performance of pFabric using its original implementation to pFabric using our approximate queue. We find that approximation has minimal effect on overall network behavior which makes performance on a mircorscale the only concern in selecting a queue for a specific scheduler.

\begin{figure}[t]
\centering
\includegraphics[width=0.95\linewidth]{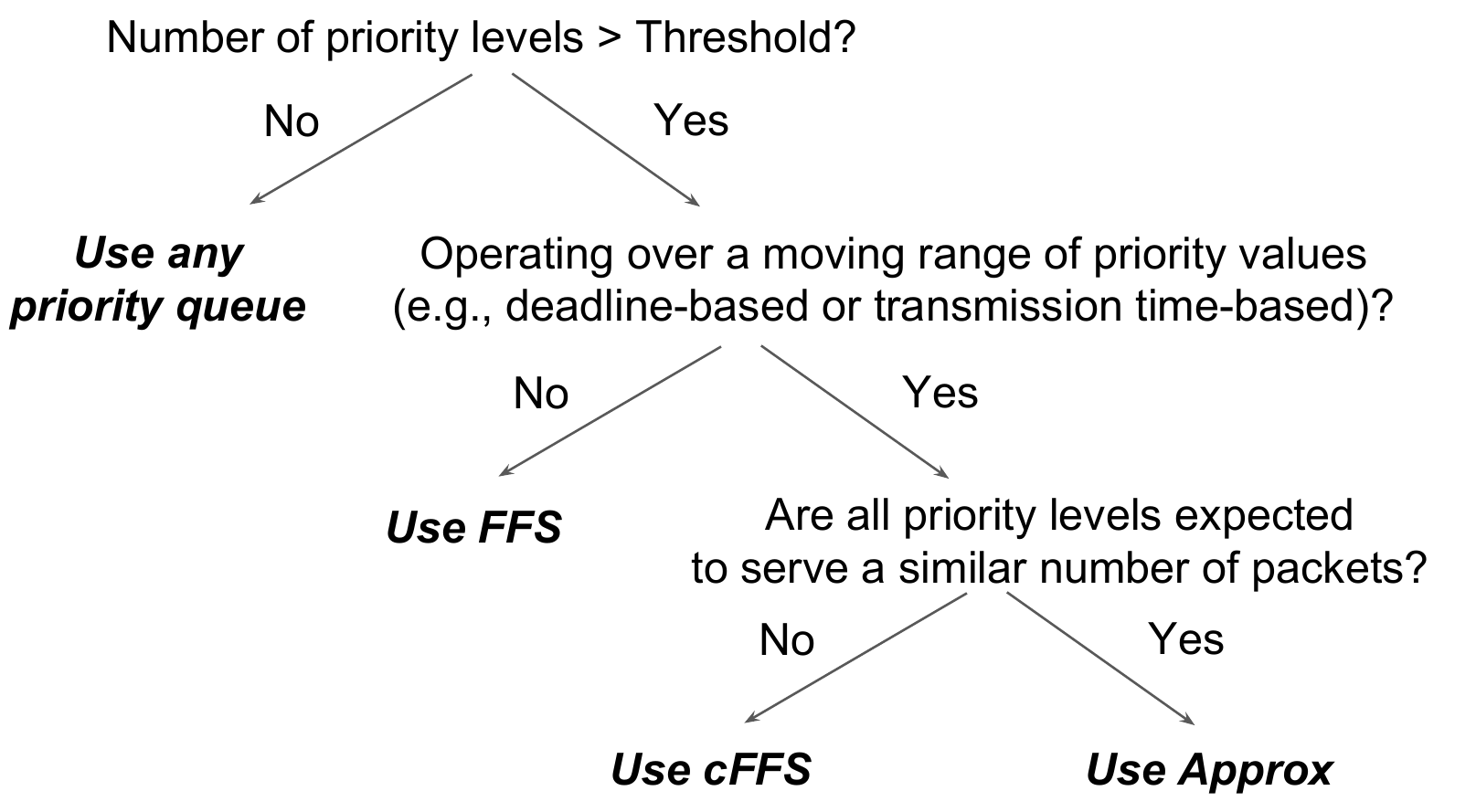}
\vspace{-0.1in}
\caption{Decision tree for selecting a priority queue based on the characteristics of the scheduling algorithm.}
\vspace{-0.1in}
\label{fig:guide}
\end{figure}

\textbf{A Guide for Choosing a Priority Queue for Packet Scheduling} Figure~\ref{fig:guide} summarizes our takeaways from working with the proposed queues. For a small number of priority levels, we find that the choice of priority queue has little impact and for most scenarios a bucket-based queue might be  overkill due to its memory overhead. However, when the number of priority levels or buckets is larger than a threshold the choice of queues makes a significant difference. We found in our experiments that this threshold is 1k and that the difference in performance is not significant around the threshold. We find that if the priority levels are over a fixed range (e.g., job remaining time \cite{pfabric}) then an FFS-based priority queue is sufficient. When the priority levels are over a moving range, where the number of levels are not all equally likely (e.g., rate limiting with a wide range of limits \cite{carousel}), it is better to use cFFS priority queue. However, for priority levels over a moving range with highly occupied priority levels (e.g., Least Slack Time-based \cite{UPS} or hierarchical-based schedules \cite{billaud2013hclock}) approximate queue can be beneficial.   

{\color{black} Another important aspect is choosing the number of buckets to assign to a queue. This parameter should be chosen based on both the desired granularity and efficiency which form a clear trade-off. Proposed queues have minimal CPU overhead (e.g., a queue with a billion buckets will require six bit operations to find the minimum non-empty bucket using a cFFS). Hence, the main source of efficiency overhead is the memory overhead which has two components: 1) memory footprint, and 2) cache freshness. However, we find that most scheduling policies require thousands to tens of thousands of elements which require small memory allocation for our proposed queues.}

\vspace{-0.1in}
{\color{black}
\section{Discussion}
\vspace{-0.1in}

\textbf{Impact of Eiffel on Next Generation Hardware Packet Schedulers:} Scheduling is widely supported in hardware switches using a fixed short list of scheduling policies, including shaping, strict priority, and Weighted Round Robin \cite{arista_1,arista_2,cisco_1,pifo}. Recently, programmable networking hardware has been moving steadily towards wide adoption and deployment over the past few years. For example, P4 \cite{Bosshart:2014:PPP:2656877.2656890} provides a data plane programming language following match action tables. Furthermore, P4 allows for setting meta data per packet which provide a way to schedule packets that relies on packet tagging to determine their relative ordering. However, P4, so far, still does not provide a programming model for packet scheduling \cite{DCP4}. PIFO \cite{pifo} proposed a hardware scheduler programming abstraction that can be integrated in a programmable data plane model. Other efforts for programmable network hardware include SmartNICs by Microsoft that leverages FPGAs to implement network logic in programmable NICs \cite{211249}.

We believe that the biggest impact Eiffel will have is making the case for a reconsideration of the basic building blocks of the packet schedulers in hardware. Current proposals for packet scheduling in hardware (e.g., PIFO model \cite{pifo} and SmartNICs \cite{211249}), rely on parallel comparisons of elements in a single queue. This approach limits the size of the queue. Earlier proposals that rely on pipelined-heaps \cite{bhagwan2000fast,4154755,wang2013per} required a priority queue that can capture the whole universe of possible packet rank values, which requires significant hardware overhead. We see Eiffel as a step on the road of improving hardware packet schedulers by reducing the number of parallel comparisons through an FFS-based queue meta data or through an approximate queue metadata. For instance, Eiffel can be employed in a hierarchical structure with parallel comparisons to increase the capacity of individual queues in a PIFO-like setting. Future programmable schedulers can implement a hardware version of cFFS or the approximate queue and provide an interface that allows for connecting them according to programmable policies. While the implementation is definitely not straight forward, we believe this to be the natural next step in the development of scalable packet schedulers.

%We believe that the biggest impact Eiffel will have on this effort is changing the building blocks of the packet schedulers in hardware. Current proposals for packet scheduling in hardware (e.g., PIFO model \cite{pifo}), rely on parallel comparisons of elements in a single queue. This approach limits the size of the queue. Earlier proposals that rely on pipelined-heaps \cite{bhagwan2000fast,4154755,wang2013per} required a priority queue that can capture the whole universe of possible packet rank values, which requires significant hardware overhead. We see Eiffel as a step on the road of improving hardware packet schedulers by reducing the number of parallel comparisons through an FFS-based queue meta data or through an approximate queue metadata. For instance, Eiffel can be employed in a hierarchical structure with parallel comparisons to increase the capacity of individual queues in a PIFO-like setting. Future programmable schedulers can implement a hardware version of cFFS or the approximate queue and provide an interface that allows for connecting them according to programmable policies. While the mapping is definitely not straight forward, we believe this to be the natural next step in the development of scalable packet schedulers for 100GbE+ networks.

\textbf{Impact of Eiffel on Network Scheduling Algorithms and Systems:} Performance and capacity of packet schedulers are key factors in designing a network-wide scheduling algorithm. It is not rare that algorithms are significantly modified to map a large number of priority values to the eight priority levels offered in IEEE 802.1Q switches \cite{1637340} (e.g., Qjump \cite{qjump}). Otherwise, algorithms are deemed unrealistic when they require priority queues with large number of priorities. We believe Eiffel can enable a reexamination of the approach to the design of such algorithms taking into account that complex scheduling can be performed using Eiffel in software in Virtual Network Functions (VNF). Hence, in settings that tolerates a bit of latency (e.g., systems already employing latency), network scheduling can be moved from the hardware to software. On the other hand, if the network is not congested, such scheduling can be performed at endhosts \cite{UPS}.
}
\vspace{-0.1in}

\section{Conclusion}
\label{sec:conc}
\vspace{-0.1in}
Efficient packet scheduling is a crucial mechanism for the correct operation of networks. Flexible packet scheduling is a necessary component of the current ecosystem of programmable networks. 
%However, scalable efficient and flexible packet scheduling has been an elusive goal with efficiency achieved by developing tailored solutions for every scheduling algorithm. 
In this paper, we showed how Eiffel can introduce both efficiency and flexibility for packet scheduling in software relying on integer priority queuing concepts and novel packet scheduling programming abstractions. We showed that Eiffel can achieve orders of magnitude improvements in performance compared to the state of the art while enabling packet scheduling at scale in terms of both number of flows or rules and line rate. We believe that our work should enable network operators to have more freedom in implementing complex policies that correspond to current networks needs where isolation and strict sharing policies are needed. Eiffel also makes the case for a reconsideration of the basic building blocks of packet schedulers which should motivate future work on schedulers implemented in hardware in NICs and switches.

{\footnotesize  \bibliography{sigproc}
 \bibliographystyle{acm}}

%{\footnotesize \bibliographystyle{acm}
%\bibliography{sigproc} }
%\theendnotes

\appendix                                     
\section{Gradient Queue Correctness}
\label{app:thm}
%\vspace{-0.08in}
\begin{theorem}
The index of the maximum non-empty bucket, $N$, is $ceil ({b}/{a}) $.
%If the index of a maximum non-empty bucket is N, then $ceil(\frac{b}{a}) =N$. % regardless of the state of the rest of the buckets.
\label{thm1}
\end{theorem}
%\vspace{-0.15in}
\begin{proof} We encode the occupancy of buckets by a bit string of length $N$ where zeros represent empty buckets and ones represent nonempty buckets. The value of the bit string is the value of the critical point $x=\frac{b}{a}$ for queue represented by the bit of strings. We prove the theorem by showing an ordering between all bit strings, where the maximum value is $N$ and the minimum value is larger than $N-1$. The minimum value is when all buckets are nonempty (i.e., all ones).  In that case, $a =\sum_{i=1}^{N}2^i$ and  $b =\sum_{i=1}^{N}i2^i$. Note that $b$ is an Arithmetic-Geometric Progression that can be simplified to $N2^{N+1}-(2^{N+1}-2)$ and $a$ is a Geometric Progression that can be simplified to $2^{N+1}-2$. Hence, the critical point $x = \frac{N2^{N+1}}{2^{N+1}-2} - 1 = \frac{N}{1-2^{-N}} - 1$ where $\frac{N}{1-2^{-N}}<N+1$ and $ceil(x) = N$. The maximum value occurs when only bucket $N$ is nonempty (i.e., all zeros). It is straightforward to show that the critical point is exactly $x=N$. Now, consider any N-bit string, where the Nth bit is 1, if we flip one bit from 1 to zero, the value of the critical point increases. It is straight forward to show that $\frac{b-j2^{j}}{a-2^{j}} - \frac{b}{a} > 0$, where $j$ is the index of the flipped bit. 
\end{proof}

\section{Examples of Errors in Approximate Gradient Queue}
\label{app:example}
To better understand the effect of missing elements on the accuracy of the approximate queue, consider the following cases of elements distribution for a maximum priority queue with $N$ buckets: 
%\vspace{-0.1in}
\begin{itemize}[leftmargin=*]
\item Elements are evenly distributed over the queue with frequency $1/\alpha$, which is equivalent to an Exact Gradient Queue with $N/\alpha$ elements,
%\vspace{-0.1in}
\item $N/2$ elements are present in buckets from $0$ to $N/2$ and then a single element is present in bucket indexed $3N/4$, where the concentration of the elements at the beginning of the queue will create an error on the estimation of the index of the maximum element  $\epsilon = ceil(b/a) + u(\alpha) - 3N/4$. We note that in this case $\epsilon < 0$ because the estimation of $ceil(b/a)$ will be closer to the concentration of elements that is pulling the curvature away from $3N/4$. The error in such cases grows proportional to size of the concentration and inversely proportional to the distance between the low concentration and the high concentration.
%\vspace{-0.1in}
\item All elements are present, which allows the value $\epsilon = ceil(b/a) + u(\alpha)$ to be exactly where the maximum element is.
\end{itemize}

%\section{State of Scheduling in Modern Networks}
%In modern networks, packet scheduling can easily become the bottleneck of the system. This is because schedulers are burdened with the overhead of maintaining a large number of buffered packets sorted according to the scheduling policies. Network operators mitigate the overhead of schedulers by using a small set of efficiently implemented scheduling policies. In particular, modern networks mostly employ shaping, strict priority, Weighted Fair Queuing, and Hierarchical Weighted Fair Queuing. Algorithms are developed to approximate each of these policies to provide efficient implementation. For instance, Deficit Round Robin (DRR) is an approximation of  General Processor Sharing (GPS) that is widely deployed in hardware \cite{pifo}. A more recent approximation is QFQ \cite{QFQ} which provides a ``very close'' approximation of GPS based on WF$^2$Q \cite{wf2q}. Another two widely deployed instances of such approximations are FQ/pacing \cite{fq} and hClock \cite{billaud2013hclock}. FQ/pacing is part of Linux's Queuing Disciplines (QDiscs) \cite{hubert2002linux} used for pacing packets while ensuring fairness through DRR. hClock is part of VMWare's ESXi hypervisor \cite{vmware} and is used as an approximation of H-GPS through Hierarchical Weighted Packet Queuing (HWPQ) \cite{hierarchicalFQ}. 

\end{document}